\documentclass[journal,twoside,twocolumn]{IEEEtranTCOM}

\normalsize

\usepackage{graphicx}
\usepackage{subfigure}
\usepackage[T1]{fontenc}
\usepackage{url}
\usepackage[cmex10]{amsmath}
\usepackage{amssymb}
\usepackage{cite}
\usepackage{amsthm}
\usepackage{textcomp}
\usepackage{siunitx}
\usepackage{color}
\usepackage{enumerate}
\usepackage{cases}
\usepackage{bm}

\newtheorem{theorem}{Theorem}
\newtheorem{lemma}{Lemma}

\newtheorem{corollary}{Corollary}

\theoremstyle{remark}
\newtheorem{remark}{Remark}

\theoremstyle{definition}
\newtheorem{definition}{Definition}

\begin{document}

\setlength{\abovedisplayskip}{1mm}
\setlength{\belowdisplayskip}{1mm}
\setlength{\abovecaptionskip}{1mm}

\sloppy

\title{Rover-to-Orbiter Communication in Mars:\\ Taking Advantage of the Varying Topology} 

\author{Songze~Li,~\IEEEmembership{Student~Member,~IEEE,}
        David~T.H.~Kao,~\IEEEmembership{Member,~IEEE,}
        and~A.~Salman~Avestimehr,~\IEEEmembership{Member,~IEEE}
\thanks{This work was partly supported by JPL R\&TD grant RSA-1517455, and NSF grants CAREER 1408639, CCF-1408755, NETS-1419632, EARS-1411244, and ONR award N000141310094. Materials in this paper were presented in part at the IEEE International Symposium on Information Theory, Hong Kong, June 2015.}
\thanks{S.~Li and A.S.~Avestimehr are with the Department of Electrical Engineering, University of Southern California, Los Angeles, CA, 90089, USA (e-mail: songzeli@usc.edu; avestimehr@ee.usc.edu).}
\thanks{This research work  was completed while D.T.H.~Kao was part of the Universtiy of Southern California (e-mail: kaod@usc.edu). D.T.H. Kao is now a part of Google Inc.}
}

\maketitle

\begin{abstract}
In this paper, we study the communication problem from rovers on Mars' surface to Mars-orbiting satellites. We first justify that, to a good extent, the rover-to-orbiter communication problem can be modelled as communication over a $2 \times 2$ X-channel with the network topology \emph{varying} over time. For such a fading X-channel where transmitters are only aware of the time-varying topology but not the time-varying channel state (i.e., no CSIT), we propose coding strategies that \emph{code across topologies}, and develop upper bounds on the sum degrees-of-freedom (DoF) that is shown to be tight under certain pattern of the topology variation. Furthermore we demonstrate that the proposed scheme approximately achieves the ergodic sum-capacity of the network. Using the proposed coding scheme, we numerically evaluate the ergodic rate gain over a time-division-multiple-access (TDMA) scheme for Rayleigh and Rice fading channels. We also numerically demonstrate that with practical orbital parameters, a 9.6\% DoF gain, as well as more than 11.6\% throughput gain can be achieved for a rover-to-orbiter communication network.  
\end{abstract}

\begin{IEEEkeywords}
Mars, Rover-to-Orbiter Communication, X-channel, Varying Topologies, Coding Across Topologies.
\end{IEEEkeywords}

\section{Introduction}
As an increasing amount of science and engineering data are collected by rovers on Mars' surface and sent back to Earth, efficient rover-to-Earth communication is of primary importance. Originally, each Mars exploration rover communicated to the deep space network (DSN) on Earth directly in a point-to-point fashion. However, the data rate suffered from limited rover transmission power, large path loss and occlusion in line-of-sight. Currently, much more data is sent to Earth via the relaying of Mars orbiting satellites (orbiters). For example, each of the 2004 Mars Exploration Rovers (MER) returned well over 150 Mbits of data to Earth per Mars solar day (sol), 92\% of which were relayed by two Mars orbiters. That was 5 times of the communication rate of the 1996 Mars Pathfinder lander, which was only capable of direct-to-earth (DTE) communication \cite{statman2004coding}. The dramatic data rate increase enabled more frequent usage of the data-rich applications like continuous sensing and monitoring, video recording and high-resolution, three-color panoramas.

Current Mars orbiters employ the decode-and-forward (DF) relaying strategies such that they decode the messages from rovers and re-encode them to send to Earth. On the rover-to-orbiter proximity link, time-division-multiple-access (TDMA) that allows one rover to talk to one orbiter at a time has been used to avoid interference. Although TDMA requires low communication overhead and little coordination between communication entities, it severely limits the overall communication rate. In this paper, we are interested in within the DF framework (although other joint-decoding schemes \cite{avestimehr2011,lim2011noisy,nazer2011} may achieve higher data rate), how can we increase the rover-to-orbiter communication rate without significantly increasing the system complexity and communication overhead? 

To answer this question, we first note that the existence of a line-of-sight (LOS) communication link from a rover to an orbiter varies over time. For example, the Mars Reconnaissance Orbiter (MRO) appears in the line-of-sight of a rover located close to the Mars north pole 12 times per sol, for a duration of 8 minutes each time~\cite{taylor2006}. Consequently, for a system comprised of multiple rovers and orbiters, the \emph{topology} of the communication network, determined by which rovers can see which orbiters, varies over time. We also note that communicating topological information to the rovers requires little communication overhead (one bit of feedback per rover-orbiter pair at each time), and is more practical than providing the rovers with full channel state information. 

Even with only knowledge of network topology at transmitters, recent results ~\cite{vahid2013,sun2013,issa2013,gherekhloo2013,GCS,chen14} have demonstrated the potential communication rate gains by \emph{coding across topologies} (CAT) for the interference channel, X-channel, and broadcast channel with varying topologies. Given the time-varying property of the rover-to-orbiter network topology, our goal in this paper is to understand the fundamental impact of coding across topologies on the performance of the rover-to-orbiter communication network. As a first step, we justify that a 2-rover, 4-orbiter communication problem can be modelled as a $2\times 2$ X-channel with varying topologies. We also assume that, due to practical constraints (e.g., limited coordination), no channel state information is available at transmitters other than the network topology (i.e., no CSIT). 

For a $2 \times 2$ Gaussian X-channel with varying topologies and no CSIT besides the network topology, we identify two \emph{coding opportunities} where we may exploit the time varying topologies by coding across them, and we characterize the sum-DoF. The key idea of the achievable scheme is to exhaustively create coding opportunities in a way that maximizes the DoF gain. Conversely, we develop upper bounds on the sum-DoF that is tight under certain pattern of the topology variation. In addition to the DoF results, we also demonstrate that our coding scheme achieves the ergodic sum-capacity to within a constant gap for uniform-phase, Rayleigh and Rice fading X-channels with varying topologies. 

Finally, we numerically demonstrate the rate gains of our proposed scheme over the rate achieved by TDMA in Rayleigh and Rice fading channels with varying topologies. We also demonstrate that for a 2-rover, 4-orbiter communication network, applying the proposed coding scheme can increase the DoF by 9.6\%, and the throughput by at least 11.6\%.

\section{Problem setting and statement of main results}
To study the rover-to-orbiter communication problem, we begin by justifying that we may model the problem as communication over a $2 \times 2$ Gaussian X-channel with varying topologies.

Currently, the Mars landing rovers (e.g., Mars Science Laboratory (MSL) Curiosity rover~\cite{makovsky2009}) send exploration data to Earth by relaying through two Mars orbiters: MRO and Mars Odyssey Orbiter~\cite{makovsky2002}. However, additional orbiters, including the Mars Express spacecraft and the recently launched Mars Atmosphere and Volatile Evolution (MAVEN) Orbiter, are also capable of relaying. We assume, according to \cite{taylor2006}, that a rover is able to establish a communication link to an orbiter if and only if the orbiter is within the line-of-sight of the rover, which is defined as:
\begin{definition}[line-of-sight]
An orbiter is said to be within the line-of-sight (LOS) of a rover if the orbiter forms at least a \ang{10} elevation angle from the surface tangent plane of the rover.  $\hfill \square$
\end{definition}

We then simulated a system comprised of 2 rovers and 4 orbiters over one sol (approximately 24 hours), and recorded the instantaneous network topologies determined by the active links. 

\begin{table}[htbp]
   \caption{Normalized time fractions of some topologies with three or more LOS links over a sol}
   \label{tab:statis}
   \centering
  \includegraphics[width=0.48\textwidth]{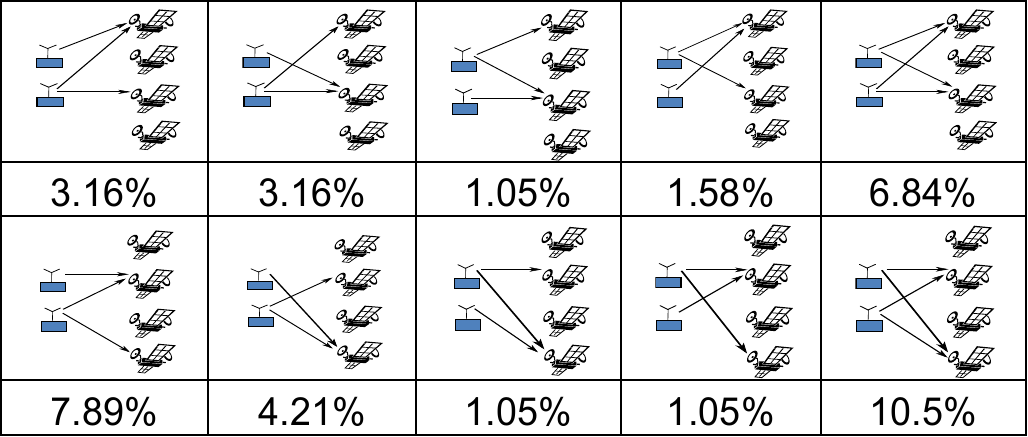}
\end{table}

Our simulation indicates that, interestingly, 40\% (sum of the time fractions of the topologies listed in Table~\ref{tab:statis}) of the time when \emph{any} communication is possible, the network topology is made up of three or more rover-to-orbiter LOS links. In all of these ``complicated'' topologies (listed in Table~\ref{tab:statis}), only 2 rovers and 2 orbiters are relevant. Furthermore, because all orbiters serve as relays to forward decoded messages to Earth, a rover's message will reach Earth as long as any single orbiter can decode the message. Since 1) the topology rarely includes more than two orbiters, 2) the identity of the rover is irrelevant, and 3) we wish to maximize the data rate, it is sufficient to study the sum-capacity of the $2 \times 2$ X-channel, where each rover communicates messages to each of the two orbiters.

Additionally, we emphasize that the network topology varies over time in a non-trivial manner. This is evidenced by observing that, during one sol, a good fraction of time (ranging from 1\% to 11\%) is spent in each of the topologies listed in Table~\ref{tab:statis}. Moreover, the time spent in each topology in Table~\ref{tab:statis} is comparable to the time spent in other topologies with only one or two LOS communication links. Therefore, to capture the time-varying properties of the topology of the aforementioned $2\times 2$ X-channel, all non-degenerate topologies should be incorporated into the problem formulation as possible scenarios. To this end, a $2 \times 2$ Gaussian X-channel with the topology varying among all possibilities (except for the case where no LOS link exists) becomes the natural choice to model the rover-to-orbiter communication problem. 

\subsection{Problem Formulation and Definitions}
A $2 \times 2$ Gaussian X-channel with varying topologies consists of two transmitters and two receivers, where each transmitter has an independent message to send to each of the two receivers. The network topology is defined by a binary vector $\mathbf{c}=(c_{11},c_{12},c_{21},c_{22})^T$, where $c_{ij}$, $i,j \in \{1,2\}$, is equal to 1 if and only if there is a LOS (defined in Definition~1) communication link from Transmitter $j$ to Receiver $i$, $c_{ij}=0$ otherwise. We index all considered topologies in Fig.~\ref{fig:top}, and let $\mathcal{A} =\{s_1,s_2,s_3,s_4,m_1,m_2,b_1,b_2,z_1,z_2,z_3,z_4,i_1,i_2,f\}$ denote the set of all topology indices. We assume that over the course of communication (consisting of $N$ time slots), the network topology changes randomly with an arbitrary distribution over the possibilities in Fig.~\ref{fig:top}, and we denote the fraction of time spent in Topology $a$ as $\lambda_a$, for all $a \in \mathcal{A}$. Because we focus only on time instances where at least one communication link exists, $\sum\limits_{a \in \mathcal{A}} \lambda_a =1$.

\begin{figure}[htbp]
   \centering
   \includegraphics[width=0.45\textwidth]{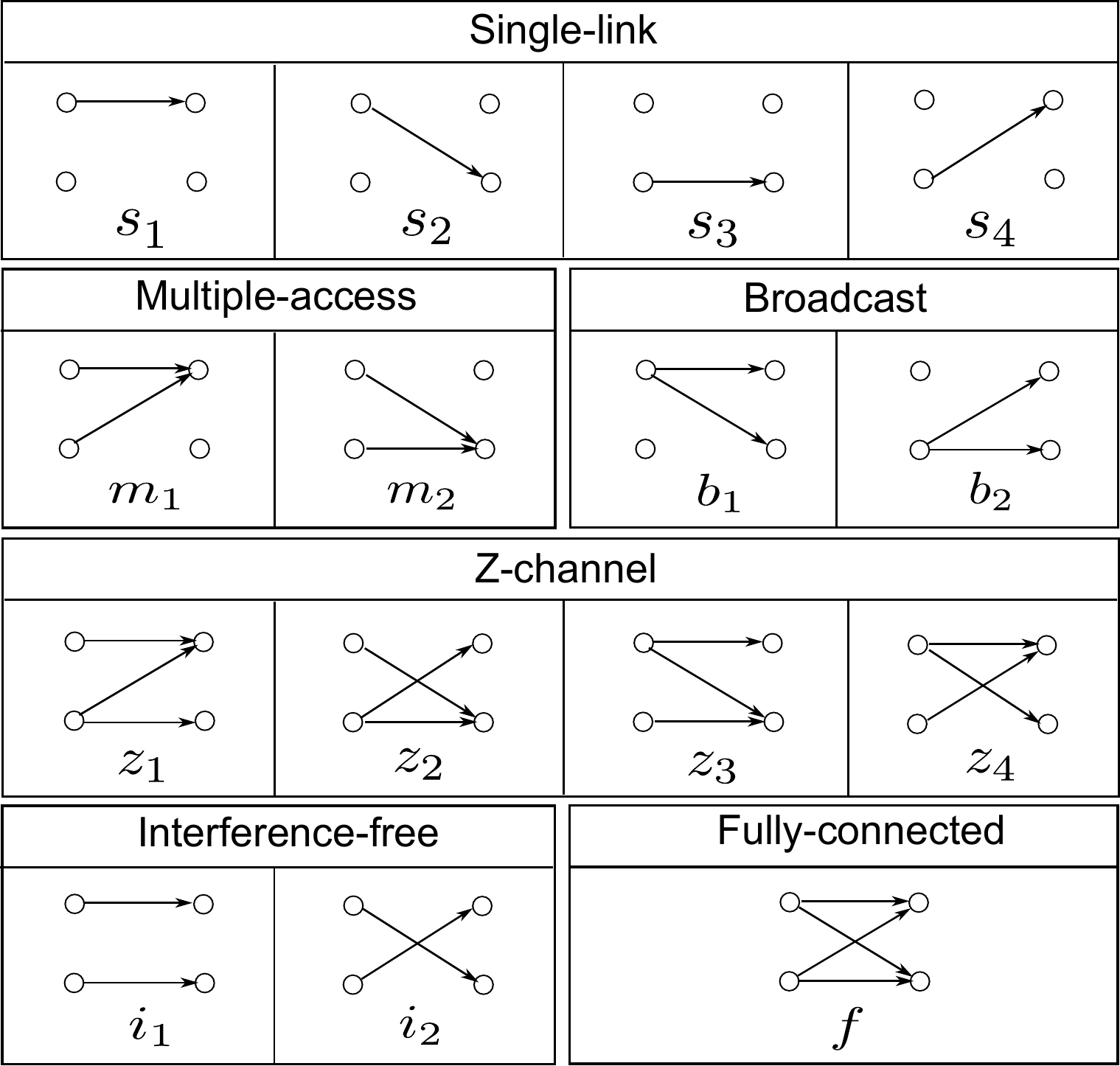}
   \caption{Possible topologies for a $2\times2$ X-channel. Each topology diagram consists of 4 nodes with 2 transmitter nodes on the left and 2 receiver nodes on the right. A directed edge from a transmitter node to a receiver node exists if and only if there is a LOS communication link from the transmitter to the receiver.}
   \label{fig:top}
\end{figure}

Given the varying topology, the output signal at Receiver $j$ in time slot $n$ is, for $i,j \in \{1,2\}$:
\begin{equation*}
Y_j(n) = \sum\limits_{i: c_{ji}(n)=1} h_{ji}(n)X_i(n) +  Z_j(n),
\end{equation*}
where $\mathbf{c}(n)$ is the topology in time slot $n$, globally known at all transmitters and receivers, $X_i(n) \in \mathbb{C}$ is the input symbol of Transmitter $i$, $h_{ji}(n) \in \mathbb{C}$ is the channel coefficient from Transmitter $i$ to Receiver $j$, and $Z_j(n)\sim \mathcal{N}(0,1)$ is the complex additive white Gaussian noise at Receiver $j$. We assume that the channel coefficients, $h_{ji}(n)$, are drawn identically and independently across indices $i,j$ from a bounded, continuous distribution. Each transmitter is subject to the same power constraint $P$, i.e., $\mathbb{E}\{|X_i(n)|^2\} \leq P$, $i = 1,2$.

We denote the collection of channel coefficients over $N$ time slots as ${{\boldsymbol h}^N = \{h_{11}(n),h_{12}(n),h_{21}(n),h_{22}(n)\}_{n=1}^N}$, and we assume that the transmitters are not aware of the values of ${\boldsymbol h}^N$, while the receivers know the values of ${\boldsymbol h}^N$ perfectly. Note that this assumption of no CSIT but network topology was also referred as \emph{minimal CSIT} in the topological interference management (TIM) problem defined in~\cite{jafar13}, .

Over $N$ time slots, message $W_{ji}\in \{1,\ldots,2^{NR_{ji}}\}$, $i,j =1,2$, is communicated from Transmitter $i$ to Receiver $j$ at rate $R_{ji}$. The rate tuple $(R_{11},R_{12},R_{21},R_{22})$ is achievable if the probability of decoding error for every message vanishes as $N \rightarrow \infty$. The capacity region $\mathcal{C}$ is the union of all achievable rate vectors, and the sum-capacity is defined as
\begin{align*}
C_{\Sigma} \overset{\Delta}{=} \underset{(R_{11},R_{12},R_{21},R_{22})\in \mathcal{C}}{\text{max}} R_{11}+R_{12}+R_{21}+R_{22}. 
\end{align*} 

We also define the sum degrees-of-freedom (DoF) as $d_{\Sigma} \overset{\Delta}{=} \lim \limits_{P \rightarrow \infty} \dfrac{C_\Sigma}{\log P}$. 

\subsection{Main Results}
Our first result characterizes the sum-DoF of the $2\times 2$ Gaussian X-channel with the topology varying symmetrically as defined below: 

\begin{definition}
A $2 \times 2$ Gaussian X-channel with the topology varying across the ones in Fig.~\ref{fig:top}, is said to have \emph{symmetric} topology variation if  $\lambda_{z_1}+\lambda_{z_4} = \lambda_{z_2}+\lambda_{z_3}$.  $\hfill \square$
\end{definition}

\begin{remark}
Notice that the property of a ``symmetric topology variation'' is a property the fractions of the time the system stays in different topologies should satisfy, but not a property of the individual topologies. $\hfill \square$
\end{remark} 

\begin{theorem}
The sum-DoF of the $2 \times 2$ Gaussian X-channel, with symmetric topology variation 
\begin{align}
d_{\Sigma} =& 1+\lambda_{i_1}+\lambda_{i_2} \nonumber \\
+& \text{min}\big\{\lambda_{z_1}+\lambda_{z_4}, \lambda_{z_1}+\lambda_{z_3}+\lambda_f, \lambda_{z_2}+\lambda_{z_4}+\lambda_f\big\}.\label{eq:DoF}
\end{align}
\end{theorem}

\begin{remark}
Other than the two interference-free topologies $i_1$ and $i_2$ (which achieves a sum-DoF of 2), all of the topologies are individually limited to a sum-DoF of 1 without CSIT~\cite{JS2008}. Hence when coding within each topology separately, the maximal achievable sum-DoF is $1+\lambda_{i_1}+\lambda_{i_2}$, and the last term in (\ref{eq:DoF}) represents the DoF gain by coding across topologies. $\hfill \square$
\end{remark}

\begin{remark}
A similar problem is solved in \cite{sun2013} with 4 admissible topologies $z_1$, $z_3$, $i_1$ and $f$. Incorporating more topologies into the problem creates new opportunities for the system to benefit from coding across topologies (see Section III). $\hfill \square$
\end{remark}

Without the symmetric constraint on the topology variation, we have the following lower and upper bounds on the sum-DoF:
\begin{theorem}
The sum-DoF of a $2 \times 2$ Gaussian X-channel, with the network topology varying over time across the topologies listed in Fig.~\ref{fig:top}, $d_{\Sigma}$ is bounded as 
\begin{align}
d_{\Sigma}  \geq & 1+\lambda_{i_1}+\lambda_{i_2}+ \min \{\lambda_{z_1}+\lambda_{z_4},\lambda_{z_2}+\lambda_{z_3}, \nonumber\\
&\lambda_{z_1}+\lambda_{z_3}+\lambda_{f}, \lambda_{z_2}+\lambda_{z_4}+\lambda_{f}\}, \label{eq:Alower}\\
d_{\Sigma}  \leq & 1+\lambda_{i_1}+\lambda_{i_2}+\min \{ \frac{\lambda_{z_1}+\lambda_{z_2}+\lambda_{z_3}+\lambda_{z_4}}{2}, \nonumber \\
&\lambda_{z_1}+\lambda_{z_3}+\lambda_{f},\lambda_{z_2}+\lambda_{z_4}+\lambda_{f} \}. \label{Aupper}
\end{align}
\end{theorem}

\begin{remark}
In general, the lower and upper bounds in Theorem~2 do not match. However, as demonstrated in Table~\ref{tab:gap} below, the average gap between the lower and the upper bounds, evaluated over 10000 randomly chosen $\{\lambda_{z_1},\lambda_{z_2},\lambda_{z_3},\lambda_{z_4},\lambda_f\}$, is rather small compared to the achieved sum-DoF in (\ref{eq:Alower}). $\hfill \square$
\end{remark}

\begin{table}[h]
	\caption{Gaps between lower and upper bounds of the sum-DoF of an X-channel with varying topologies}
	\label{tab:gap}
	\centering
    \begin{tabular}{| c | c | c |}
    \hline
    $\lambda_{z_1}\!+\!\lambda_{z_2}\!+\!\lambda_{z_3}\!+\!\lambda_{z_4}\!+\!\lambda_f$ & Maximum Gap & Average Gap \\ \hline
   0.2 & 0.0887 & 0.0174 \\ \hline
   0.5 & 0.2164 & 0.0436 \\ \hline
   0.8 & 0.3515 & 0.0702 \\ \hline
   \end{tabular}
\end{table}

\begin{remark}
Using the proposed coding scheme in Theorem~2 (described in Section~III and Appendix~II) on the simulated system whose topology variation is specified in Table~I achieves a DoF gain of 9.47\% over the TDMA scheme. To see this:
\begin{enumerate}
\item We apply the proposed coding scheme respectively to the X-channel containing the 5 topologies at the top of Table~I and another X-channel containing the 5 topologies at the bottom of Table~I.
\item Then we use the normalized time fractions in Table~I to calculate the DoF gains over TDMA (i.e., the last term in~(\ref{eq:Alower})) in each of the X-channels.  That is 4.21\% for the top X-channel and 5.26\% for the bottom X-channel.
\item We sum up these two DoF gains and compare with the DoF achieved by the TDMA scheme (=1 because $\lambda_{i_1}+\lambda_{i_2} = 0$ for both X-channels in the simulation results) to obtain an overall DoF gain of 9.47\%. $\hfill \square$
\end{enumerate} 
\end{remark}

When channel coefficients are also i.i.d. varying over time, the following result characterizes the ergodic sum-rate achieved by our proposed coding scheme:   
\begin{theorem}
For the $2 \times 2$ Gaussian X-channel with symmetric topology variation (defined in Definition~2) and i.i.d. (across space and time) channel coefficients, we define ${\phi \overset{\Delta}{=} \text{min}\{|\lambda_{z_1}\!-\!\lambda_{z_2}|, \lambda_f\}}$, and an ergodic sum-rate $R_\Sigma$ is achievable if
\begin{align}
R_{\Sigma} \leq & A\left[\sum \limits_{k=1}^4 \lambda_{s_k} + 2(\lambda_{i_1}+\lambda_{i_2})+\lambda_{b_1}+\lambda_{b_2}\right] \nonumber \\
&+ \!B \big[\lambda_{m_1}\!+\!\lambda_{m_2}\!+\!|\lambda_{z_1}\!-\!\lambda_{z_2}|\!+\!|\lambda_{z_3}\!-\!\lambda_{z_4}|\!+\!\lambda_f -3\phi \big] \nonumber\\
&+ \!C\left[\text{min}\{\lambda_{z_1},\lambda_{z_2}\} + \text{min}\{\lambda_{z_3},\lambda_{z_4}\}\right] +D\phi,
\end{align} 
where 
\begin{align*}
A \!=& \mathbb{E}\left\{\log(1+|h_{11}|^2P)\right\}, \\
B \!=& \mathbb{E}\left\{\log(1+|h_{11}|^2P+|h_{12}|^2P)\right\},\\
C \!=& \mathbb{E}\left\{\log\left(\left|\mathbf{I} + P \mathbf{V}\mathbf{V}^H\right|\right)\right\} \\
&+ \mathbb{E}\left\{\log\left(1+\frac{|h_{21}(2)|^2 P}{1+|h_{22}(2)|^2\slash |h_{22}(1)|^2}\right)\right\},\\
D \!=& 2\mathbb{E}\left \{\! \log \!\! \left(\frac{\Bigg | \begin{bmatrix}1 & 0\\ 0 & 1+|h_{12}(3)|^2 \slash |h_{12}(1)|^2 
\end{bmatrix}\!\!+\! P \mathbf{U}\mathbf{U}^H\Bigg |}{1+|h_{12}(3)|^2 \slash |h_{12}(1)|^2}\right)\! \right\},\\
\mathbf{V} \!=& \begin{bmatrix}
h_{11}(1) & h_{12}(1)\\ 0 & h_{12}(2)
\end{bmatrix},\quad
\mathbf{U} \!= \begin{bmatrix}
h_{12}(2) & h_{11}(2) \\ 0 & h_{11}(3) \end{bmatrix},
\end{align*}
with all expectations taken over random channel coefficients.
\end{theorem}

Furthermore, we demonstrate via the following corollary of Theorem~3 that our coding scheme approximately achieves the ergodic sum-capacity of the $2\times 2$ Gaussian X-channel with symmetric topology variation, for uniform-phase, Rayleigh and Rice fading channel coefficients.

\begin{corollary}
The ergodic sum-capacity of the $2\times 2$ Gaussian X-channel with symmetric topology variation and i.i.d. (across space and time) channel coefficients is achievable to within a constant gap for one of the following channel distributions if $\lambda_{z_1} = \lambda_{z_2}$ or $\lambda_f \leq |\lambda_{z_1}-\lambda_{z_2}|$:
\begin{itemize}
\vspace{-1.5 mm}
\item $h_{ik} = e^{j\theta_{ik}}$, and $\theta_{ik}$ is uniformly distributed over $[0,2\pi)$, $i,k = 1,2$,
\item Rayleigh fading with $\mathbb{E}\{h_{ik}\}=0$ and $\mathbb{E}\{|h_{ik}|^2\}=1$, $i,k=1,2$,
\item Rice fading with Rice factor $K_r=1$.
\end{itemize} 
\end{corollary}

We present the achievable scheme of Theorem~1 in Section~III and demonstrate that the proposed scheme is indeed optimal by proving the converse of Theorem~1 in Section~IV and Appendix~I. In Appendix~II we provide a sketch of the proof of Theorem~2, which extends the results of Theorem~1 to the general asymmetric topology variation setting. Finally, Theorem~3 and Corollary~1 are proved in Section~V to demonstrate the rate performance of the proposed scheme for fading channels.
 
\section{Achievability of Theorem 1}\label{sec:achieve}
A \emph{coding opportunity} is a combination of topologies, coding across which achieves a DoF gain over coding in each topology separately. In this section, we first identify two coding opportunities for the $2 \times 2$ X-channel with varying topologies, and then state our achievable scheme, which optimally creates and exploits the two coding opportunities to achieve the sum-DoF of Theorem~1. For ease of exposition, we will ignore noise terms in the DoF analysis. We also note that, by assumption, all channel coefficients are non-zero with probability 1 at any time.

\subsection{Coding Opportunity 1: $\{z_1,z_2\}$, $\{z_3,z_4\}$}\label{sec:coding1}
We illustrate in Fig.~\ref{fig:CO1} a two-phase transmission strategy for the topology combination $\{z_1,z_2\}$ to achieve a sum-DoF of $\frac{3}{2}$, and note that the same scheme can be applied to topology combination $\{z_3,z_4\}$ by relabelling the transmitters and receivers. 

\begin{figure}[htbp]
   \centering
   \includegraphics[width=0.45\textwidth]{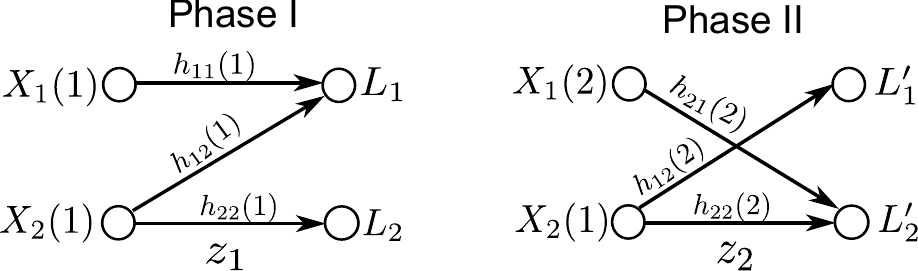}
   \caption{Illustration of a $\frac{3}{2}$-DoF coding scheme for topology combination $\{z_1,z_2\}$. Receiver 1 decodes $X_1(1)$ and $X_2(1)$ from two linearly independent combinations $(L_1,L_1')$ of $(X_1(1),X_2(1))$, and Receiver 2 decodes $X_1(2)$ from two linearly independent combinations $(L_2,L_2')$ of $(X_1(2),X_2(1))$.}
   \label{fig:CO1}
\end{figure}

Receiver 1 uses $\left(L_1,L_1'\right)$ to decode $X_1(1)$ and $X_2(1)$ (always decodable because all channel gains are almost surely non-zero); Receiver 2 uses $L_2$ to cancel $X_2(1)$ from $L_2'$, and then decodes $X_1(2)$. In total, 3 symbols are decoded in 2 time slots, achieving a sum-DoF of $\frac{3}{2}$. 

\subsection{Coding Opportunity 2: $\{z_2, z_4, f\}$, $\{z_1, z_3, f\}$} \label{sec:coding2}
Next, we illustrate in Fig.~\ref{fig:CO2} a three-phase transmission strategy for the topology combination $\{z_2,z_4,f\}$ (and $\{z_1,z_3,f\}$ by relabelling the receivers) to achieve a sum-DoF of $\frac{4}{3}$. 

Receiver 1 uses $L_1$ to cancel $X_2(1)$ from $L_1''$, obtaining a residual signal $\tilde{L}_1''$, and decodes $X_1(2)$ and $X_2(2)$ using $L_1'$ and $\tilde{L}_1''$. Receiver 2 uses $L_2'$ to cancel $X_1(2)$ from $L_2''$, obtaining a residual signal $\tilde{L}_2''$, and decodes $X_{1}(1)$ and $X_2(1)$ using $L_2$ and $\tilde{L}_2''$. Overall, 4 symbols are decoded in 3 time slots, achieving a sum-DoF of $\frac{4}{3}$. 

\begin{figure}[htbp]
   \centering
   \includegraphics[width=0.45\textwidth]{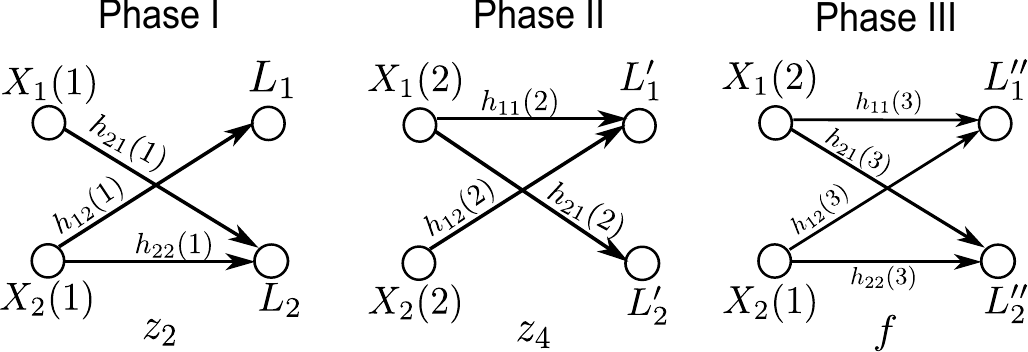}
   \caption{Illustration of a $\frac{4}{3}$-DoF coding scheme for topology combination $\{z_2,z_4,f\}$. Receiver 1 decodes $X_1(2)$ and $X_2(2)$ from 3 linearly independent combinations $(L_1,L_1',L_1'')$ of $(X_2(1),X_1(2),X_2(2))$, and Receiver 2 decodes $X_1(1)$ and $X_2(1)$ from 3 linearly independent combinations $(L_2,L_2',L_2'')$ of $(X_1(1),X_2(1),X_1(2))$.}
   \label{fig:CO2}
\end{figure}

\begin{remark}
The scheme for this coding opportunity was introduced for a binary fading interference channel, in which the channel links are either ``on'' or ``off'' \cite{vahid2013}. The same concept was used in \cite{sun2013} to improve the sum-DoF of the two-user interference channel and X-channel with alternating topology. Finally in \cite{issa2013}, the authors demonstrated that for the two-hop interference channel, the best vector-linear strategy at the intermediate relays is to vary relaying coefficients over  time, creating end-to-end topologies equivalent to the three used in coding opportunity 2. $\hfill \square$
\end{remark}

\subsection{Coding Scheme for $2 \times 2$ X-channel with Varying Topologies}
Equipped with the transmission strategies for the two coding opportunities, we proceed to present our achievable scheme that focuses on the symmetric topology variation as defined in Definition~2. The general idea of the achievable scheme is to create and utilize the two coding opportunities in a greedy manner to maximize the overall sum-DoF.

First, we code in all instances of interference-free topologies $i_1$ and $i_2$ separately at DoF $2$, delivering a total of $2(\lambda_{i_1}+\lambda_{i_2})N$ symbols. Then we split the description of the rest of the coding scheme into the following two cases:

\noindent \emph{Case 1: $\lambda_{z_1} \leq \lambda_{z_2}$ and $\lambda_{z_3} \leq \lambda_{z_4}$}
\begin{enumerate}[i.]
\item We pair each $z_1 (z_3)$ instance with a $z_2(z_4)$ instance until $z_1 (z_3)$ is exhausted, obtaining $(\lambda_{z_1}+\lambda_{z_3})N$ instances of coding opportunity 1. For each paired $\{z_1,z_2\}$, we perform the \mbox{$\frac{3}{2}$-DoF} scheme in Section~\ref{sec:coding1}, delivering $3\lambda_1N$ symbols. Similarly, another $3\lambda_{z_3}N$ symbols are delivered by performing coding opportunity 1 scheme on all $\{z_3,z_4\}$ pairs.   

\item We pair one of  $(\lambda_{z_2}-\lambda_{z_1})N$ surplus $z_2$ instances with one of  $(\lambda_{z_4}-\lambda_{z_3})N$ surplus $z_4$ instances, and one of $\lambda_fN$ $f$ instances until one of them is exhausted, obtaining $\theta N$ instances of coding opportunity 2, where $\theta \overset{\Delta}{=}\text{min}\{\lambda_{z_2}-\lambda_{z_1},\lambda_{z_4}-\lambda_{z_3},\lambda_f\}$. Applying the $\frac{4}{3}$-DoF scheme in Section~\ref{sec:coding2} to each coding opportunity 2 instance, we deliver $4\theta N$ symbols.

\item Lastly, for each of $(\lambda_{z_2}\!-\!\lambda_{z_1}\!-\theta)N$ surplus $z_2$ instances, $(\lambda_{z_4}\!-\!\lambda_{z_3}\!-\!\theta)N$ surplus $z_4$ instances, $(\lambda_{f}\!-\!\theta)N$ surplus $f$ instances, and all the remaining topology instances, we code them separately, delivering 1 symbol in each of them.     
\end{enumerate}

Thus we achieve the overall sum-DoF:
\begin{align}
&3(\lambda_{z_1}\!+\!\lambda_{z_3}) \!+\! 4\theta \!+\! (\lambda_{z_2}\!-\!\lambda_{z_1}\!-\!\theta) \!+\! (\lambda_{z_4}\!-\!\lambda_{z_3}\!-\!\theta) \!+\! (\lambda_f \!-\!\theta) \nonumber \\
&+\! \sum\limits_{k=1}^{4}\lambda_{s_k} \!+\! \sum\limits_{k=1}^{2}(\lambda_{m_k}\!+\!\lambda_{b_k}\!+\! 2\lambda_{i_k})\nonumber\\
=& 1+ \lambda_{i_1}+\lambda_{i_2}+\text{min}\{\lambda_{z_2}+\lambda_{z_3},\lambda_{z_1}+\lambda_{z_3}+\lambda_f\}.
\end{align}

\begin{figure}[htbp]
   \centering
   \includegraphics[width=0.48\textwidth]{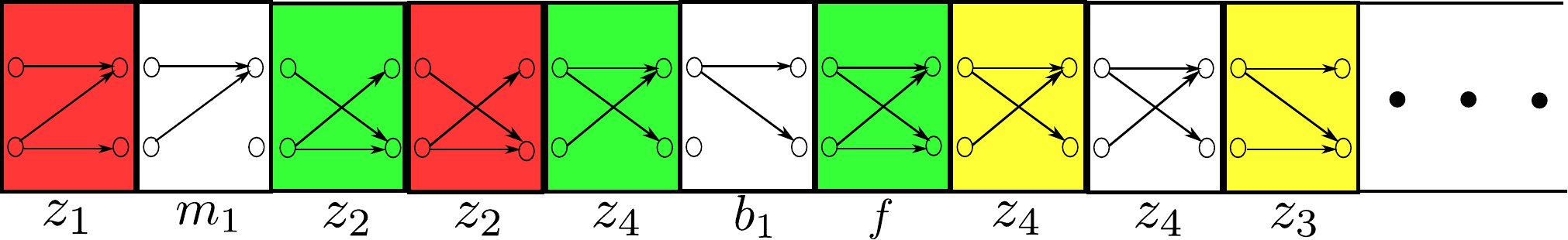}
   \caption{Coding scheme applied to an example realization of time-varying network topologies for Case 1. Coding opportunity 1 is colored red for $\{z_1,z_2\}$, and yellow for $\{z_3,z_4\}$; Coding opportunity 2 is colored green, and all  topologies coded separately are left blank. Notice that in this example, no more future topology instance has Topology $z_1$ or $z_3$.}
   \label{fig:coding_str1}
\end{figure}

\noindent \emph{Case 2: $\lambda_{z_1} > \lambda_{z_2}$ and $\lambda_{z_3} > \lambda_{z_4}$}
The achievable scheme is similar to the first case, whereas topology instances of $z_1$ and $z_3$ are surplus from Step i. and the instances of combination $\{z_1,z_3,f\}$ are exhaustively consumed using the scheme for coding opportunity 2. The achieved sum-DoF in this case is $1\!+\! \lambda_{i_1} \!+\! \lambda_{i_2}\!+\!\text{min}\{\lambda_{z_1}+\lambda_{z_4},\lambda_{z_2}+\lambda_{z_4}+\lambda_f\}$.

Therefore in general, the achieved sum-DoF when the topology varies symmetrically (${\lambda_{z_1}+\lambda_{z_4} = \lambda_{z_2}+\lambda_{z_3}}$), $d_{\Sigma}$ is 
\begin{align}
 &1\! + \! \lambda_{i_1} \!+ \!\lambda_{i_2} \nonumber\\
 &+\!\text{min}\bigg\{ \lambda_{z_1}\!+\!\lambda_{z_4},\text{min}\{\lambda_{z_1}\!,\lambda_{z_2}\}\!+\!\text{min}\{\lambda_{z_3}\!,\lambda_{z_4}\}\!+\!\lambda_f  \bigg\}\nonumber\\
=& 1 \!+\! \lambda_{i_1}\!+\!\lambda_{i_2} \!+\! \text{min}\big\{\!\lambda_{z_1}\!+\!\lambda_{z_4},\lambda_{z_1}\!+\!\lambda_{z_3}\!+\!\lambda_f, \lambda_{z_2}\!+\!\lambda_{z_4}\!+\!\lambda_f \! \big\}.\label{eq:achieved_DoF}
\end{align}

\begin{remark}
As mentioned in Remark~2, when the variation of the topology is symmetric, the overall DoF gain by coding across topologies is $\min\big\{\!\lambda_{z_1}\!+\!\lambda_{z_4},\lambda_{z_1}\!+\!\lambda_{z_3}\!+\!\lambda_f, \lambda_{z_2}\!+\!\lambda_{z_4}\!+\!\lambda_f \! \big\}$. From the above performance analysis of our proposed scheme, we find that the DoF gain due to coding opportunity 1 is $\min\{\lambda_{z_1},\lambda_{z_2}\}+\min\{\lambda_{z_3},\lambda_{z_4}\} = \min\{\lambda_{z_1}+\lambda_{z_3},\lambda_{z_2}+\lambda_{z_4}\}$, and the DoF gain due to coding opportunity 2 is $\min\{|\lambda_{z_1}-\lambda_{z_2}|,|\lambda_{z_3}-\lambda_{z_4}|,\lambda_f\}$. $\hfill \square$
\end{remark}

\section{Converse of Theorem 1}
To prove the optimality of the achieved sum-DoF in (\ref{eq:achieved_DoF}), we will show in this section that the sum-DoF $d_{\Sigma}$ is bounded as
\begin{align}
d_{\Sigma} & \leq 1+ \lambda_{i_1} + \lambda_{i_2}+\lambda_{z_1}+\lambda_{z_4},\label{eq:bond1}\\
d_{\Sigma} & \leq 1+ \lambda_{i_1} + \lambda_{i_2}+\lambda_{z_1}+\lambda_{z_3}+\lambda_f,\label{eq:bond2}\\
d_{\Sigma} & \leq 1+ \lambda_{i_1} + \lambda_{i_2}+\lambda_{z_2}+\lambda_{z_4}+\lambda_f.\label{eq:bond3}
\end{align} 

We point out that taking the minimum of these three bounds results in the expression in (\ref{eq:achieved_DoF}), therefore proving these three bounds proves the converse. We will present the proof of (\ref{eq:bond1}) in this section and defer the proofs for (\ref{eq:bond2}) and (\ref{eq:bond3}) to Appendix~I.

We define the length-$N$ output vector at Receiver $j$ as $Y_j^{N}$, $j = 1,2$, and the sub-vector of $Y_{j}^N$ received when the network is in Topology $a$ as $Y_{j,a}^{N}$. The received vector at Receiver 1 may thus be expressed as $Y_{1}^N = \big(\{Y_{1,s_k}^N\}_{k=1}^4,Y_{1,m_1}^N, Y_{1,m_2}^N, Y_{1,b_1}^N, Y_{1,b_2}^N, \{Y_{1,z_k}^N\}_{k=1}^4,Y_{1,i_1}^N,$ $Y_{1,i_2}^N, Y_{1,f}^N\big)$.

Further, for any subset of topology indices, $\mathcal{S}\subseteq\mathcal{A}$, $Y_{j,\mathcal{S}}^{N} \overset{\Delta}{=} \{Y_{j,a}^N: a \in \mathcal{S}\}$. 
In particular, we define the set of topologies $\mathcal{B}\subseteq\mathcal{A}$ as $\mathcal{B} \overset{\Delta}{=} \mathcal{A} \backslash \{s_1,s_2,s_3,s_4,m_1,m_2\}$.

Before beginning proofs of the individual bounds, we make one last general observation that in the two broadcast topologies $b_1$, $b_2$ and the fully-connected topology $f$, because the transmitters only know the network topology, and the channel coefficients are i.i.d. across space, $Y^N_{1,a}$ and $Y^N_{2,a}$, are statistically equivalent for $a \in \mathcal{E} \triangleq\{b_1,b_2,f\}$. Hence, if a receiver can decode the desired messages using the signals of $Y_{2,\mathcal{E}}^N$, it can also do so using $Y_{1,\mathcal{E}}^N$. That is, without any rate loss, Receiver 2 should be able to decode $(W_{21},W_{22})$ using $Y_{2}^N \!\!=\!\! \left(\{Y_{2,s_k}^N\}_{k=1}^4, Y_{2,m_1}^N, Y_{2,m_2}^N,\{Y_{2,z_k}^N\}_{k=1}^4, Y_{2,i_1}^N, Y_{2,i_2}^N, Y_{1,\mathcal{E}}^N\right)$.

We start the proof of the bound in (\ref{eq:bond1}) with the following two steps:

\noindent\emph{1. Channel Enhancements: }
We enhance the system without any rate loss by allowing the two transmitters to cooperate at infinite rate. As a result, the system is converted into a two-user MISO broadcast channel with message $W_j$ (with rate $R_j$) intended for Receiver $j$, $j=1,2$. 

\noindent\emph{2. Virtual Receivers: }
We introduce two virtual receivers, labeled $\tilde{1}$ and $\tilde{2}$, which are statistically indistinguishable from Receiver~1 and 2 respectively. 

Specifically at Receiver $\tilde{1}$, $Y_{\tilde{1},a} =Y_{1,a}$ when $a \neq z_1,z_4$. When the topology at time $n$ is $z_1$ or $z_4$,
\begin{equation}\label{eq:outz1z4}
\begin{aligned}
Y_1(n) &= h_{11}(n)X_1(n) + h_{12}(n)X_2(n) + Z_1(n),\\
Y_{\tilde{1}}(n) &= h_{\tilde{1}1}(n)X_1(n) + h_{\tilde{1}2}(n)X_2(n) + Z_{\tilde{1}}(n).
\end{aligned}
\end{equation}

At Receiver $\tilde{2}$, $Y_{\tilde{2},a}=Y_{2,a}$ when $a \neq z_2,z_3$. When the topology at time $n$ is $z_2$ or $z_3$,
\begin{equation}\label{eq:outz2z3}
\begin{aligned}
Y_2(n) &= h_{21}(n)X_1(n) + h_{22}(n)X_2(n) + Z_2(n),\\
Y_{\tilde{2}}(n) &= h_{\tilde{2}1}(n)X_1(n) + h_{\tilde{2}2}(n)X_2(n) + Z_{\tilde{2}}(n).
\end{aligned}
\end{equation}

In (\ref{eq:outz1z4}) and (\ref{eq:outz2z3}), for $j \in \{1,2\}$, $(h_{\tilde{j}1}(n), h_{\tilde{j}2}(n))$ is identically distributed and independent of $(h_{j1}(n),h_{j2}(n))$, and $Z_j(n)$ and $Z_{\tilde{j}}(n)$ are i.i.d. Gaussian with zero mean and unit variance.

We denote the collection of channel coefficients over $N$ time slots (including the channel coefficients at the virtual receivers) as $\tilde{{\boldsymbol h}}^N \overset{\Delta}{=} \left\{h_{ji}(n): j \in \{1,2,\tilde{1},\tilde{2}\}, i \in \{1,2\}\right\}_{n=1}^N$, and  recall the assumption that the values of $\tilde{{\boldsymbol h}}^N$ are not known at the transmitters, but are perfectly known at Receiver $j$, $j=1,2,\tilde{1},\tilde{2}$. We also note that due to the statistical equivalence of $Y_j^N$ and $Y_{\tilde{j}}^N$, if $R_j$ is achievable at Receiver $j$, then Receiver $\tilde{j}$ can also decode $W_j$.

\begin{lemma}
In Topologies $z_1$ and $z_4$ ($z_2$ and $z_3$), given the signals at Receivers $1$ and $\tilde{1}$ ($2$ and $\tilde{2}$) and the channel coefficients, the signals at Receiver $2$ ($1$) can be approximately reconstructed:
\begin{align}
h(Y_{2,a}^N|Y_{1,a}^N,Y_{\tilde{1},a}^N,\tilde{{\boldsymbol h}}^N) &\leq No(\log P), \; a = z_1, z_4, \label{eq:recon2}\\
h(Y_{1,a}^N|Y_{2,a}^N,Y_{\tilde{2},a}^N,\tilde{{\boldsymbol h}}^N ) &\leq No(\log P), \; a = z_2, z_3.\label{eq:recon1}
\end{align} 
\end{lemma}

\begin{proof}
We prove the case when the topology is $z_1$ ($a=z_1$ in (\ref{eq:recon2})). The proofs for the other 3 cases follow similarly, and are omitted here. 

We recall that if the topology at time $n$ is $z_1$, then by (\ref{eq:outz1z4}), $Y_1(n)$ and $Y_{\tilde{1}}(n)$ almost surely provide two noisy linearly independent equations of $X_1(n)$ and $X_2(n)$. More specifically,
\begin{equation*}
\begin{bmatrix}
Y_1(n) \\ Y_{\tilde{1}}(n) \\ Y_2(n)
\end{bmatrix} = \begin{bmatrix}
h_{11}(n) & h_{12}(n) \\ h_{\tilde{1}1}(n) & h_{\tilde{1}2}(n) \\ 0 & h_{22}(n)
\end{bmatrix} \begin{bmatrix}
X_1(n) \\X_2(n) 
\end{bmatrix} + \begin{bmatrix}
Z_1(n) \\ Z_{\tilde{1}}(n) \\Z_2(n)
\end{bmatrix},
\end{equation*}
where $\begin{bmatrix}
h_{11}(n) & h_{12}(n) \\ h_{\tilde{1}1}(n) & h_{\tilde{1}2}(n)
\end{bmatrix}$ is invertible almost surely. 

\begin{align*}
&h(Y_2(n)|Y_1(n),Y_{\tilde{1}}(n),\tilde{{\boldsymbol h}}^N\!) \\
= & h \bigg(\!\!Y_2(n)\!-\!\!\begin{bmatrix}
0 & h_{22}(n)
\end{bmatrix}\!\!\!\begin{bmatrix}
h_{11}(n) & h_{12}(n) \\ h_{\tilde{1}1}(n) & h_{\tilde{1}2}(n)
\end{bmatrix}^{-1}\!\!\begin{bmatrix}
Y_1(n) \\ Y_{\tilde{1}}(n)
\end{bmatrix} \!\! \Bigg| \\
&Y_1(n),\!Y_{\tilde{1}}(n),\tilde{{\boldsymbol h}}^N \!\! \bigg)\\
\leq & h \bigg(\!\! Z_2(n)\!-\!\begin{bmatrix}
0 & h_{22}(n)
\end{bmatrix}\!\!\begin{bmatrix}
h_{11}(n) & h_{12}(n) \\ h_{\tilde{1}1}(n) & h_{\tilde{1}2}(n)
\end{bmatrix}^{-1}\!\!\begin{bmatrix}
Z_1(n) \\ Z_{\tilde{1}}(n)
\end{bmatrix}\!\!\Bigg|\tilde{{\boldsymbol h}}^N \!\!\bigg)\\
=& o(\log P),
\end{align*}

and thus
\begin{align*}
h(Y_{2,z_1}^N|Y_{1,z_1}^N,Y_{\tilde{1},z_1}^N,\tilde{{\boldsymbol h}}^N\!)&\! \leq \!\!\!\!\!\!\!\sum \limits_{\substack{n \in \text{ time slots}\\\text{in Topology }  z_1}} \!\!\!\!\!\! h(Y_2(n)|Y_1(n),Y_{\tilde{1}}(n),\tilde{{\boldsymbol h}}^N\! )\\
&\leq No(\log P).
\end{align*}
\end{proof}

In the rest of the proof, we first derive two upper bounds on $R_1$ at Receiver $1$ and $\tilde{1}$ respectively, then combine the two bounds and use Lemma 1 to establish an upper bound on $2R_1$. We then repeat the same steps at Receiver $2$ and $\tilde{2}$ to obtain an upper bound on $2R_2$. Lastly, we sum up the upper bounds for $2R_1$ and $2R_2$ to arrive at an upper bound for the sum-rate, yielding (\ref{eq:bond1}) as an upper bound for the sum-DoF.

At Receiver 1, we apply Fano's inequality while treating $\tilde{{\boldsymbol h}}^N$ as a part of the output signals:
\begin{align}
N&R_1  \overset{(a)}{\leq} I(W_1;Y_1^N|\tilde{{\boldsymbol h}}^N) + N\epsilon_N \nonumber\\
 =& I(W_1;Y_{1,\mathcal{B}}^N|\tilde{{\boldsymbol h}}^N) \nonumber \\
 &+\! I(W_1;\{Y_{1,s_k}^N\}_{k=1}^4, Y_{1,m_1}^N, Y_{1,m_2}^N|Y_{1,\mathcal{B}}^N,\tilde{{\boldsymbol h}}^N)\!+\! N\epsilon_N  \nonumber\\
= & I(W_1;Y_{1,\mathcal{B}}^N|\tilde{{\boldsymbol h}}^N) \!+\! h(\{Y_{1,s_k}^N\}_{k=1}^4, Y_{1,m_1}^N, Y_{1,m_2}^N|Y_{1,\mathcal{B}}^N,\tilde{{\boldsymbol h}}^N)\nonumber\\
&-\! h(\{Y_{1,s_k}^N\}_{k=1}^4, Y_{1,m_1}^N, Y_{1,m_2}^N|W_1,Y_{1,\mathcal{B}}^N,\tilde{{\boldsymbol h}}^N)\!+\! N\epsilon_N\\
\overset{(b)}{\leq} &I(W_1;Y_{1,\mathcal{B}}^N|\tilde{{\boldsymbol h}}^N) \!+\! \sum\limits_{k=1}^4 h(Y_{1,s_k}^N|\tilde{{\boldsymbol h}}^N) \nonumber \\
&+\! h(Y_{1,m_1}^N|\tilde{{\boldsymbol h}}^N)\!+\! h(Y_{1,m_2}^N|\tilde{{\boldsymbol h}}^N)\!+\! N\epsilon_N \label{eq:BCR1}\\
\leq & I(W_1;Y_{1,\mathcal{B}}^N|\tilde{{\boldsymbol h}}^N) \nonumber \\
&+ N(\lambda_{s_1}+\lambda_{s_4}+\lambda_{m_1})(\log P+o(\log P)) + N\epsilon_N,\label{eq:fano1}
\end{align}
where (a) holds since $\tilde{{\boldsymbol h}}^N$ is independent of $W_1$, and (b) is because that given that $\left(X_1^N,X_2^N\right)$ is a function of $(W_1,W_2)$
\begin{align*}
&h(\{Y_{1,s_k}^N\}_{k=1}^4,Y_{1,m_1}^N,Y_{1,m_2}^N|W_1,Y_{1,\mathcal{B}}^N,\tilde{{\boldsymbol h}}^N) \\
\geq& h(\{Y_{1,s_k}^N\}_{k=1}^4,Y_{1,m_1}^N,Y_{1,m_2}^N|W_1,W_2,Y_{1,\mathcal{B}}^N,\tilde{{\boldsymbol h}}^N)\\
 =& \! h(\{Y_{1,s_k}^N\}_{k=1}^4,Y_{1,m_1}^N,Y_{1,m_2}^N|W_1,W_2,X_1^N,X_2^N,Y_{1,\mathcal{B}}^N,\tilde{{\boldsymbol h}}^N)\\
=&\! N\left(\sum \limits_{k=1}^4\lambda_{s_k}+\lambda_{m_1}+\lambda_{m_2}\right)\log(\pi e) \geq 0.
\end{align*}

Similarly at Receiver $\tilde{1}$,
\begin{align}
NR_1 \leq & I(W_1;Y_{\tilde{1},\mathcal{B}}^N|\tilde{{\boldsymbol h}}^N) \nonumber\\
&+ N(\lambda_{s_1}+\lambda_{s_4}+\lambda_{m_1})(\log P+o(\log P))+N\epsilon_N. \label{eq:fano1t}
\end{align}

Combining (\ref{eq:fano1}) and (\ref{eq:fano1t}), we have
\begin{align}
2NR_1 \leq & I(W_1;Y_{1,\mathcal{B}}^N|\tilde{{\boldsymbol h}}^N)\!+\! I(W_1;Y_{\tilde{1},\mathcal{B}}^N|\tilde{{\boldsymbol h}}^N) \nonumber\\
&+\! 2N(\lambda_{s_1}\!+\!\lambda_{s_4}\!+\!\lambda_{m_1})(\log P\!+\!o(\log P))\!+\! N\epsilon_N.\label{eq:R1}
\end{align}

We can further bound the terms $I(W_1;Y_{1,\mathcal{B}}^N|\tilde{{\boldsymbol h}}^N)$ and $I(W_1;Y_{\tilde{1},\mathcal{B}}^N|\tilde{{\boldsymbol h}}^N)$ in (\ref{eq:R1}) as follows:
\begin{align}
&I(W_1;Y_{1,\mathcal{B}}^N|\tilde{{\boldsymbol h}}^N) \nonumber \\
=& I(W_1;\{Y_{1,z_k}^N\}_{k=1}^4,\! Y_{1,\mathcal{E}}^N|\tilde{{\boldsymbol h}}^N) \nonumber \\
&+\!I(W_1;Y_{1,i_1}^N,\! Y_{1,i_2}^N|\{Y_{1,z_k}^N\}_{k=1}^4,\! Y_{1,\mathcal{E}}^N,\tilde{{\boldsymbol h}}^N)\nonumber\\
\leq &I(W_1;\{Y_{1,z_k}^N\}_{k=1}^4,Y_{1,\mathcal{E}}^N|\tilde{{\boldsymbol h}}^N) \nonumber \\
&+ h(Y_{1,i_1}^N,\! Y_{1,i_2}^N|\{Y_{1,z_k}^N\}_{k=1}^4,\! Y_{1,\mathcal{E}}^N,\tilde{{\boldsymbol h}}^N)\\
\leq & I(W_1;Y_{1,z_2}^N,Y_{1,z_3}^N,Y_{1,\mathcal{E}}^N|\tilde{{\boldsymbol h}}^N\!) \!\!+\!\! N(\lambda_{i_1}\!\!+\!\lambda_{i_2})(\log P \!+\! o(\log P))\nonumber \\
&+I(W_1;Y_{1,z_1}^N,\!Y_{1,z_4}^N|Y_{1,z_2}^N,Y_{1,z_3}^N,Y_{1,\mathcal{E}}^N,\tilde{{\boldsymbol h}}^N) \\
\leq & I(W_1;Y_{1,z_2}^N,Y_{1,z_3}^N,Y_{1,\mathcal{E}}^N|\tilde{{\boldsymbol h}}^N\!)  \!+\!\! N(\lambda_{i_1}\!\!+\!\lambda_{i_2})(\log P \!\!+\!\! o(\log P)) \nonumber \\
&+ h(Y_{1,z_1}^N,\!Y_{1,z_4}^N|Y_{1,z_2}^N,Y_{1,z_3}^N,Y_{1,\mathcal{E}}^N,\tilde{{\boldsymbol h}}^N)\nonumber\\
&-\!h(Y_{1,z_1}^N,\!Y_{1,z_4}^N|W_1, \!Y_{1,z_2}^N\!, Y_{1,z_3}^N\!, Y_{1,\mathcal{E}}^N,\tilde{{\boldsymbol h}}^N)\\
\leq & N(\lambda_{i_1}\!+\!\lambda_{i_2}\!+\!\lambda_{z_1}\!+\!\lambda_{z_4})(\log P \!+\! o(\log P))\nonumber \\
&+I(W_1;Y_{1,z_2}^N,Y_{1,z_3}^N,Y_{1,\mathcal{E}}^N|\tilde{{\boldsymbol h}}^N)\nonumber \\ &-h(Y_{1,z_1}^N,Y_{1,z_4}^N|W_1,Y_{1,z_2}^N,Y_{1,z_3}^N,Y_{1,\mathcal{E}}^N,\tilde{{\boldsymbol h}}^N),\label{eq:w1y1b}
\end{align}

and similarly,
\begin{align}
&I(W_1;Y_{\tilde{1},\mathcal{B}}^N|\tilde{{\boldsymbol h}}^N) \nonumber \\
\leq & N(\lambda_{i_1}\!+\!\lambda_{i_2}\!+\!\lambda_{z_1}\!+\!\lambda_{z_4})(\log P \!+\! o(\log P)) \nonumber \\
&+\! I(W_1;Y_{1,z_2}^N,Y_{1,z_3}^N,Y_{1,\mathcal{E}}^N|\tilde{{\boldsymbol h}}^N)\nonumber \\
&-h(Y_{\tilde{1},z_1}^N,Y_{\tilde{1},z_4}^N|W_1,Y_{1,z_2}^N,Y_{1,z_3}^N,Y_{1,\mathcal{E}}^N,\tilde{{\boldsymbol h}}^N).\label{eq:w1ty1b}
\end{align}

Summing up the two bounds in (\ref{eq:w1y1b}) and (\ref{eq:w1ty1b}), we find
\begin{align}
&I (W_1;Y_{1,\mathcal{B}}^N|\tilde{{\boldsymbol h}}^N)+ I(W_1;Y_{\tilde{1},\mathcal{B}}^N|\tilde{{\boldsymbol h}}^N)\nonumber\\
\leq & 2N(\lambda_{i_1}+\lambda_{i_2}+\lambda_{z_1}+\lambda_{z_4})(\log P + o(\log P)) \nonumber \\
&+2I(W_1;Y_{1,z_2}^N,Y_{1,z_3}^N,Y_{1,\mathcal{E}}^N|\tilde{{\boldsymbol h}}^N)\nonumber \\
&-h(Y_{1,z_1}^N,Y_{1,z_4}^N,Y_{\tilde{1},z_1}^N,Y_{\tilde{1},z_4}^N|W_1,Y_{1,z_2}^N,Y_{1,z_3}^N,Y_{1,\mathcal{E}}^N,\tilde{{\boldsymbol h}}^N)\\
\leq & 2N(\lambda_{i_1}+\lambda_{i_2}+\lambda_{z_1}+\lambda_{z_4})(\log P + o(\log P)) \nonumber \\
&+ 2I(W_1;Y_{1,z_2}^N,Y_{1,z_3}^N,Y_{1,\mathcal{E}}^N|\tilde{{\boldsymbol h}}^N) \nonumber \\
&+h(Y_{1,z_2}^N,Y_{1,z_3}^N,Y_{1,\mathcal{E}}^N|W_1,\tilde{{\boldsymbol h}}^N)\nonumber \\
&-h(Y_{1,z_1}^N,Y_{1,z_4}^N,Y_{\tilde{1},z_1}^N,Y_{\tilde{1},z_4}^N,Y_{1,\mathcal{E}}^N|W_1,\tilde{{\boldsymbol h}}^N)\\
\overset{(a)}{\leq} & 2N(\lambda_{i_1}+\lambda_{i_2}+\lambda_{z_1}+\lambda_{z_4})(\log P + o(\log P)) \nonumber \\
&+ 2I(W_1;Y_{1,z_2}^N,Y_{1,z_3}^N,Y_{1,\mathcal{E}}^N|\tilde{{\boldsymbol h}}^N) \nonumber \\
&+ h(Y_{1,z_2}^N,Y_{1,z_3}^N,Y_{1,\mathcal{E}}^N|W_1,\tilde{{\boldsymbol h}}^N) \nonumber \\
&-h(Y_{1,z_1}^N,Y_{1,z_4}^N,Y_{\tilde{1},z_1}^N,Y_{\tilde{1},z_4}^N,Y_{1,\mathcal{E}}^N|W_1,\tilde{{\boldsymbol h}}^N)\nonumber\\
& + h(Y_{1,z_1}^N,Y_{1,z_4}^N,Y_{\tilde{1},z_1}^N,Y_{\tilde{1},z_4}^N,Y_{1,\mathcal{E}}^N|W_1,W_2,\tilde{{\boldsymbol h}}^N) \nonumber \\
&-h(Y_{1,z_2}^N,Y_{1,z_3}^N,Y_{1,\mathcal{E}}^N|W_1,W_2,\tilde{{\boldsymbol h}}^N)\\
= & 2N(\lambda_{i_1}+\lambda_{i_2}+\lambda_{z_1}+\lambda_{z_4})(\log P + o(\log P)) \nonumber \\
&+ 2I (W_1;Y_{1,z_2}^N,Y_{1,z_3}^N,Y_{1,\mathcal{E}}^N|\tilde{{\boldsymbol h}}^N) \nonumber \\
&+ I(W_2;Y_{1,z_2}^N,Y_{1,z_3}^N,Y_{1,\mathcal{E}}^N|W_1,\tilde{{\boldsymbol h}}^N)\nonumber \\
&-I(W_2;Y_{1,z_1}^N,Y_{1,z_4}^N,Y_{\tilde{1},z_1}^N,Y_{\tilde{1},z_4}^N,Y_{1,\mathcal{E}}^N|W_1,\tilde{{\boldsymbol h}}^N)\\
= & 2N(\lambda_{i_1}+\lambda_{i_2}+\lambda_{z_1}+\lambda_{z_4})(\log P + o(\log P)) \nonumber \\
&+ I(W_1,W_2;Y_{1,z_2}^N,Y_{1,z_3}^N,Y_{1,\mathcal{E}}^N|\tilde{{\boldsymbol h}}^N) \nonumber \\
& + I(W_1;Y_{1,z_2}^N,Y_{1,z_3}^N,Y_{1,\mathcal{E}}^N|\tilde{{\boldsymbol h}}^N) \nonumber \\
&- I(W_2;Y_{1,z_1}^N,Y_{1,z_4}^N,Y_{\tilde{1},z_1}^N,Y_{\tilde{1},z_4}^N,Y_{1,\mathcal{E}}^N|W_1,\tilde{{\boldsymbol h}}^N)\label{eq:EUint3}\\
\leq & 2N(\lambda_{i_1}+\lambda_{i_2}+\lambda_{z_1}+\lambda_{z_4})(\log P + o(\log P)) \nonumber\\
&+ h(Y_{1,z_2}^N,Y_{1,z_3}^N,Y_{1,\mathcal{E}}^N|\tilde{{\boldsymbol h}}^N) + I(W_1;Y_{1,z_2}^N,Y_{1,z_3}^N,Y_{1,\mathcal{E}}^N|\tilde{{\boldsymbol h}}^N) \nonumber \\
&- I(W_2;Y_{1,z_1}^N,Y_{1,z_4}^N,Y_{\tilde{1},z_1}^N,Y_{\tilde{1},z_4}^N,Y_{1,\mathcal{E}}^N|W_1,\tilde{{\boldsymbol h}}^N) \\
\leq &N\big[2(\lambda_{i_1}\!+\!\lambda_{i_2}\!+\!\lambda_{z_1}\!+\!\lambda_{z_4})\!+\!\lambda_{z_2}\!+\!\lambda_{z_3}\!+\!\lambda_{b_1}\!+\!\lambda_{b_2}\!+\!\lambda_f\big]\nonumber \\
&(\log P + o(\log P))+I\left(W_1;Y_{1,z_2}^N,Y_{1,z_3}^N,Y_{1,\mathcal{E}}^N\right|\tilde{{\boldsymbol h}}^N) \nonumber \\
&\!\!-\!\! I(W_2;Y_{1,z_1}^N,Y_{1,z_4}^N,Y_{\tilde{1},z_1}^N,Y_{\tilde{1},z_4}^N,Y_{2,z_1}^N,Y_{2,z_4}^N,Y_{1,\mathcal{E}}^N|W_1,\tilde{{\boldsymbol h}}^N) \nonumber\\
&\!\!+\!\!\underbrace{I(W_2;Y_{2,z_1}^N,Y_{2,z_4}^N|W_1,Y_{1,z_1}^N,Y_{1,z_4}^N,Y_{\tilde{1},z_1}^N,Y_{\tilde{1},z_4}^N,Y_{1,\mathcal{E}}^N,\tilde{{\boldsymbol h}}^N\!)}_{\leq h\left(Y_{2,z_1}^N|Y_{1,z_1}^N,Y_{\tilde{1},z_1}^N,\tilde{{\boldsymbol h}}^N\right) + h\left(Y_{2,z_4}^N|Y_{1,z_4}^N,Y_{\tilde{1},z_4}^N,\tilde{{\boldsymbol h}}^N\right)},\label{eq:lemma1}
\end{align}
where Step (a) results from the observations: 
\begin{align*}
&h(Y_{1,z_1}^N,Y_{1,z_4}^N,Y_{\tilde{1},z_1}^N,Y_{\tilde{1},z_4}^N,Y_{1,\mathcal{E}}^N|W_1,W_2,\tilde{{\boldsymbol h}}^N) \nonumber \\
=& N\left[2(\lambda_{z_1}+\lambda_{z_4})+\lambda_{b_1}+\lambda_{b_2}+\lambda_f\right]\log(\pi e),\\
&h(Y_{1,z_2}^N,Y_{1,z_3}^N,Y_{1,\mathcal{E}}^N|W_1,W_2,\tilde{{\boldsymbol h}}^N) \nonumber \\
=& N\left[\lambda_{z_2}+\lambda_{z_3}+\lambda_{b_1}+\lambda_{b_2}+\lambda_f\right]\log(\pi e),
\end{align*}
and $\lambda_{z_1} +\lambda_{z_4} =\lambda_{z_2}+\lambda_{z_3}$, due to symmetry of the topology variation.

Then by virtue of Lemma 1, (\ref{eq:lemma1}) becomes
\begin{align}
&I (W_1;Y_{1,\mathcal{B}}^N|\tilde{{\boldsymbol h}}^N)+ I(W_1;Y_{\tilde{1},\mathcal{B}}^N|\tilde{{\boldsymbol h}}^N)\nonumber\\
\leq &N\left[2(\lambda_{i_1}\!+\!\lambda_{i_2}\!+\!\lambda_{z_1}\!+\!\lambda_{z_4})\!+\!\lambda_{z_2}\!+\!\lambda_{z_3}\!+\!\lambda_{b_1}\!+\!\lambda_{b_2}\!+\!\lambda_f\right]\nonumber \\
&\left(\log P + o(\log P)\right)+I(W_1;Y_{1,z_2}^N,Y_{1,z_3}^N,Y_{1,\mathcal{E}}^N | \tilde{{\boldsymbol h}}^N) \nonumber \\
&-I(W_2;Y_{2,z_1}^N,Y_{2,z_4}^N,Y_{1,\mathcal{E}}^N|W_1,\tilde{{\boldsymbol h}}^N )+ N o(\log P)\nonumber\\
\overset{(b)}{\leq} & N\left[2(\lambda_{i_1}+\lambda_{i_2})+3(\lambda_{z_1}+\lambda_{z_4})+\lambda_{b_1}+\lambda_{b_2}+\lambda_{f}\right]\nonumber \\
&\left(\log P+o(\log P)\right)+ I(W_1;Y_{1,z_2}^N,Y_{1,z_3}^N,Y_{1,\mathcal{E}}^N|W_2,\tilde{{\boldsymbol h}}^N )\nonumber \\
&- I(W_2;Y_{2,z_1}^N, Y_{2,z_4}^N,Y_{1,\mathcal{E}}^N|W_1,\tilde{{\boldsymbol h}}^N )+No(\log P),\label{eq:w1y1bw1y1t}
\vspace{-2 mm}
\end{align}
where Step (b) is due to the symmetric topology variation, and that $W_1$ and $W_2$ are independent.
\vspace{-1 mm}
Similarly for Receivers 2 and $\tilde{2}$, we have  
\begin{align}
2NR_2 
\leq & I(W_2;Y^N_{1,\mathcal{E}},Y_{2,\mathcal{B}\backslash\mathcal{E}}^N|\tilde{{\boldsymbol h}}^N )+ I(W_2;Y^N_{1,\mathcal{E}},Y_{\tilde{2},\mathcal{B}\backslash\mathcal{E}}^N|\tilde{{\boldsymbol h}}^N )\nonumber\\
& \!+\! 2N(\lambda_{s_2}+\lambda_{s_3}+\lambda_{m_2})(\log P+o(\log P)) \!+\! N\epsilon_N, \label{eq:R2}
\end{align}
where
\begin{align}
I&(W_2;Y^N_{1,\mathcal{E}},Y_{2,\mathcal{B}\backslash\mathcal{E}}^N|\tilde{{\boldsymbol h}}^N )+ I(W_2;Y^N_{1,\mathcal{E}},Y_{\tilde{2},\mathcal{B}\backslash\mathcal{E}}^N|\tilde{{\boldsymbol h}}^N )\nonumber\\
\leq & N\left[2(\lambda_{i_1}+\lambda_{i_2})+3(\lambda_{z_1}+\lambda_{z_4})+\lambda_{b_1}+\lambda_{b_2}+\lambda_{f}\right]\nonumber \\
&\left(\log P+o(\log P)\right) \nonumber + I(W_2;Y_{2,z_1}^N, Y_{2,z_4}^N,Y_{1,\mathcal{E}}^N|W_1,\tilde{{\boldsymbol h}}^N )\nonumber \\
&- I(W_1;Y_{1,z_2}^N,Y_{1,z_3}^N,Y_{1,\mathcal{E}}^N|W_2,\tilde{{\boldsymbol h}}^N )+No(\log P).\label{eq:w2y1e}
\end{align}

Adding both sides of (\ref{eq:R1}) and (\ref{eq:R2}), and by (\ref{eq:w1y1bw1y1t}) and (\ref{eq:w2y1e}), we arrive at
\begin{align}
2&N(R_1+R_2)\nonumber \\
\leq & 2N\bigg[2(\lambda_{i_1}+\lambda_{i_2})+3(\lambda_{z_1}+\lambda_{z_4})+\lambda_{b_1}+\lambda_{b_2}+\lambda_{f}\bigg]\nonumber \\
&(\log P+o(\log P))\!+\!No(\log P)\!+\! N \epsilon_N \nonumber\\
& +2N\left(\sum\limits_{k=1}^4 \lambda_{s_k} + \lambda_{m_1} + \lambda_{m_2}\right)(\log P+o(\log P))  \nonumber\\
=& 2N\left(1+\lambda_{i_1}+\lambda_{i_2}+\lambda_{z_1}+\lambda_{z_4}\right)(\log P+o(\log P))\nonumber \\
&+No(\log P)+N \epsilon_N\label{eq:sum1}.
\end{align}

The last equality is a result of evaluating the total probability of topological states: $\sum\limits_{k=1}^4 \lambda_{s_k}+\sum\limits_{k=1}^2(\lambda_{m_k}+\lambda_{b_k}+\lambda_{i_k})+2(\lambda_{z_1}+\lambda_{z_4})+\lambda_f =1$. Dividing both sides of (\ref{eq:sum1}) by $2N \log P$ and let both $N$ and $P$ go to infinity, we obtain the sum-DoF upper bound in (\ref{eq:bond1}).

\begin{remark}
In \cite{chen14} and \cite{tandon2013}, similar techniques were utilized to upper bound the DoF in the context of fading MISO broadcast channel with alternating (varying) CSIT.  $\hfill \square$
\end{remark}

\section{Ergodic sum-capacity of $2 \times 2$ Gaussian X-channel with symmetric topology variation}
In this section, we explore the performance of our CAT-based scheme at finite SNR. We start with the proof of Theorem~3, by establishing an ergodically achievable sum-rate for the $2 \times 2$ Gaussian X-channel with symmetric topology variation.
We then demonstrate the approximate optimality of our coding scheme by proving Corollary~1. 

\subsection{Achievability (Proof of Theorem~3)}
The achievable scheme follows the approach described in Section~III, which maximizes the sum-DoF. Before proceeding to present the ergodic sum-rate achieved by individual topologies and coding opportunities, we recall the assumption that channel coefficients are i.i.d. across space and time, transmitters know the network topology but not the actual values of the channel coefficients, and the receivers know both topology and all of the channel coefficients perfectly. The rate expressions $A$, $B$, $C$, $D$ are as defined in the statement of Theorem~3.

\subsubsection{Topologies that jointly achieve $A$} Across all instances of Topologies $s_1$, $s_4$, $b_1$, $b_2$, $s_2$ and $s_3$, using point-to-point codes, one of the two transmitters communicates to Receiver~1 in $s_1$, $s_4$, $b_1$ and $b_2$, and to Receiver $2$ in $s_2$ and $s_3$, achieving an ergodic rate equal to $A$. In all instances of Topologies $i_1$ and $i_2$, each transmitter communicates with the receiver towards which it has a communication link, achieving an ergodic sum-rate $2A$.  

\subsubsection{Topologies that jointly achieve $B$} 
Across all instances of Topologies $m_1$ and $z_1$ (or $z_4$), and/or $f$ that were surplus from creating coding opportunity 1 or 2, transmitters use Gaussian codebooks with rates dictated by the ergodic sum-capacity of the fading multiple access channel from the two transmitters to Receiver~1, achieving a sum-rate $B$. Similarly, $B$ can be achieved across instances of Topology $m_2$ and surplus $z_2$ (or $z_3$).

\subsubsection{Coding Opportunity 1 achieves $C$}
Over instances of combination $\{z_1,z_2\}$ (same for $\{z_3,z_4\}$), we encode 3 sub-messages, each using a Gaussian codebook, and employ the two-phase scheme described in Section~\ref{sec:coding1}. 

The received signals in the first time slot are:
\begin{equation}
\begin{aligned}
Y_1(1) &= h_{11}(1)X_1(1)+h_{12}(1)X_2(1)+Z_1(1),  \\
Y_2(1) &=h_{22}(1)X_2(1)+Z_2(1). 
\end{aligned}
\end{equation} 

The received signals in the second time slot are:
\begin{equation}
\begin{aligned}
Y_1(2) &=\! h_{12}(2)X_2(1)\!+\!Z_1(2), \\
Y_2(2) &=\! h_{21}(2)X_1(2)\!+\!h_{22}(2)X_2(1)\!+\!Z_2(2). 
\end{aligned}
\end{equation} 

With many instances of $\left[Y_1(1) \;\; Y_1(2)\right]^T$, Receiver~1 uses MIMO decoding of the codewords associated with $X_1(1)$ and $X_2(1)$ to achieve a sum-rate of $\mathbb{E}\left\{\log\left(\left|\mathbf{I} + P \mathbf{V} \mathbf{V}^H\right|\right)\right\}$. 

On the other hand, in \emph{each} instance of coding opportunity 1, Receiver~2 uses $Y_2(1)$ to cancel $X_2(1)$ from $Y_2(2)$, resulting in the residual signal ${h_{21}(2)X_1(2) \!+\! Z_2(2)\!-\!\dfrac{h_{22}(2)}{h_{22}(1)}Z_2(1)}$. Using the residual signals, Receiver~2 can decode the sub-message associated with $X_1(2)$ at rate $\mathbb{E}\left\{\log\left(1+\dfrac{|h_{21}(2)|^2 P}{1+|h_{22}(2)|^2\slash |h_{22}(1)|^2}\right)\right\}$.

Adding up achieved rates at two receivers, we obtain the sum-rate of $C$ bits per instance of coding opportunity 1, as desired.

\subsubsection{Coding Opportunity 2 achieves $D$}
Over the instances of topology combination $\{z_2,z_4,f\}$ (same for $\{z_1,z_3,f\}$), we create random Gaussian codebooks for 4 sub-messages, and apply the three-phase scheme in Section~\ref{sec:coding2}. 

The received signals in the first time slot are:
\begin{equation}
\begin{aligned}
Y_1(1) &= h_{12}(1)X_2(1)\!+\!Z_1(1), \\
Y_2(1) &= h_{21}(1)X_1(1)\!+\! h_{22}(1)X_2(1)\!+\!Z_2(1).  
\end{aligned}
\end{equation} 
The received signals in the second time slot are:
\begin{equation}
\begin{aligned}
Y_1(2) &=\! h_{11}(2)X_1(2)\!+\!h_{12}(2)X_2(2)\!+\!Z_1(2), \\
Y_2(2) &=\! h_{21}(2)X_1(2)\!+\!Z_2(2). 
\end{aligned}
\end{equation} 
The received signals in the third time slot are:
\begin{equation}
\begin{aligned}
Y_1(3) &=\! h_{11}(3)X_1(2) \!+\! h_{12}(3)X_2(1)\!+\!Z_1(3), \\
Y_2(3) &=\! h_{21}(3)X_1(2) \!+\! h_{22}(3)X_2(1)\!+\!Z_2(3). 
\end{aligned}
\end{equation}

In each instance of coding opportunity 2, Receiver 1 uses $Y_1(1)$ to completely cancel $X_2(1)$ from $Y_1(3)$, yielding a residual signal ${\tilde{Y}_1(3) = h_{11}(3)X_1(2)+Z_1(3)\!-\!\frac{h_{12}(3)}{h_{12}(1)}Z_1(1)}$. Performing MIMO decoding on many instances of $\left[Y_1(2) \; \; \tilde{Y}_1(3)\right]^T$, Receiver 1 can decode the sub-messages associated with $X_1(2)$ and $X_2(2)$ at rate $r \!=\!\mathbb{E}\left\{\!\log\!\!\left(\!\frac{\Bigg|\begin{bmatrix}1 & 0\\ 0 & 1+|h_{12}(3)|^2 \slash |h_{12}(1)|^2 
\end{bmatrix}\!+\! P \mathbf{U}\mathbf{U}^H\Bigg|}{1+|h_{12}(3)|^2 \slash |h_{12}(1)|^2}\right)\!\!\right\}$.

Receiver 2 uses $Y_2(2)$ to completely cancel $X_1(2)$ from $Y_2(3)$, and decodes the sub-messages associated with $X_1(1)$ and $X_2(1)$ with the same ergodic rate at Receiver 1.

Therefore, an ergodic sum-rate $2r \!=\! D$ bits per coding opportunity 2 instance is achieved.

\subsubsection*{Overall Ergodic Sum-Rate}
In $N$ time slots, our proposed coding scheme creates $\text{min}\{\lambda_{z_1},\lambda_{z_2}\}N + \text{min}\{\lambda_{z_3},\lambda_{z_4}\}N$ instances of coding opportunity 1, and $\text{min}\{|\lambda_{z_1}-\lambda_{z_2}|,\lambda_f\}N$ instances of coding opportunity 2. Let $\phi \overset{\Delta}{=}\text{min}\{|\lambda_{z_1}-\lambda_{z_2}|,\lambda_f\}$, we also have $\left(|\lambda_{z_1}-\lambda_{z_2}|-\phi\right)N$ instances of Topology $z_1$ (or $z_2$), $\left(|\lambda_{z_3}-\lambda_{z_4}|-\phi\right)N$ instances of Topology $z_3$ (or $z_4$), and $(\lambda_f-\phi)N$ instances of Topology $f$ surplus from creating the two coding opportunities. Therefore, applying the coding scheme described above, we achieve the ergodic sum-rate $R_{\Sigma}$:
\begin{align}
R_{\Sigma} =& A\left[\sum \limits_{k=1}^4 \lambda_{s_k} + 2(\lambda_{i_1}+\lambda_{i_2})+\lambda_{b_1}+\lambda_{b_2}\right] \nonumber \\
&+ \!B \left[\lambda_{m_1}\!+\!\lambda_{m_2}\!+\!|\lambda_{z_1}\!-\!\lambda_{z_2}|\!+\!|\lambda_{z_3}\!-\!\lambda_{z_4}|\!+\!\lambda_f -3\phi \right] \nonumber\\
&+ \!C\left[\text{min}\{\lambda_{z_1},\lambda_{z_2}\} + \text{min}\{\lambda_{z_3},\lambda_{z_4}\}\right] +D\phi,\label{eq:ergodic_rate}
\end{align} 

\subsection{Proof of Corollary 1}
The proof proceeds in three steps. First, we develop 3 upper bounds on the ergodic sum-capacity. Second, we compare the upper bounds with the achievable rate $R_{\Sigma}$ in (\ref{eq:ergodic_rate}) to identify the gaps between them. Lastly, we evaluate the gaps with uniform-phase, Rayleigh and Rice faded channel coefficients respectively. 

\subsubsection{Upper Bounds}
We begin by proving 3 upper bounds $U_1$, $U_2$ and $U_3$ on the ergodic sum-capacity $C_{\Sigma}$ such that $C_{\Sigma} \leq \text{min}\{U_1,U_2,U_3\}$. The derivations of $U_1$, $U_2$ and $U_3$ are carried out along the same line of reasoning as in the proofs of (\ref{eq:bond1}), (\ref{eq:bond2}) and (\ref{eq:bond3}) respectively in the converse for Theorem~1. The key difference is that we carefully evaluate or bound all $o(\log P)$ terms, and then take expectation over random channel coefficients to arrive at:   
\begin{align}
U_1 =& A\left[\sum \limits_{k=1}^4 \lambda_{s_k} + 2(\lambda_{i_1}+\lambda_{i_2})+ \lambda_{b_1} +\lambda_{b_2}+\lambda_{z_1}+\lambda_{z_4}\right]\nonumber \\
&+B\left[\lambda_{m_1} + \lambda_{m_2} + \sum \limits_{k=1}^4 \lambda_{z_k} + \lambda_f\right] + E(\lambda_{z_1}+\lambda_{z_4}) \nonumber \\
&+\log(\pi e)\Bigg[\sum\limits_{k=1}^4{\lambda_{s_k}}+\lambda_{m_1}+\lambda_{m_2} \nonumber \\
&+\! 2(\lambda_{i_1}+\lambda_{i_2})+2(\lambda_{z_1}+\lambda_{z_4})\Bigg], \label{eq:EU1}\\
U_2 = &A\left[\sum \limits_{k=1}^4 \lambda_{s_k} + 2(\lambda_{i_1}+\lambda_{i_2})+ \lambda_{b_1} +\lambda_{b_2}+\lambda_{z_1}+\lambda_{z_3} \right]\nonumber \\
&+ B\left[\lambda_{m_1} + \lambda_{m_2} + \sum \limits_{k=1}^4 \lambda_{z_k} +2\lambda_f\right] + F(\lambda_{z_2}+\lambda_{z_4}) \nonumber \\
&+ \log(\pi e) \left[1+\lambda_{i_1}\!+\! \lambda_{i_2}\!+\! 2(\lambda_{z_1}\!+\! \lambda_{z_4}) \!+\! \lambda_f\right], \label{eq:EU2}\\
U_3 =& A\left[\sum \limits_{k=1}^4 \lambda_{s_k} + 2(\lambda_{i_1}+\lambda_{i_2})+ \lambda_{b_1} +\lambda_{b_2}+\lambda_{z_2}+\lambda_{z_4} \right] \nonumber \\
&+ B\left[\lambda_{m_1} + \lambda_{m_2} + \sum \limits_{k=1}^4 \lambda_{z_k} +2\lambda_f\right] + F(\lambda_{z_1}+\lambda_{z_3})\nonumber \\
&+ \log(\pi e) \left[1+\lambda_{i_1}+\lambda_{i_2}+2(\lambda_{z_1}\!+\!\lambda_{z_4}) \!+\! \lambda_f\right],\label{eq:EU3}\\
\text{where}\nonumber \\
E \!\!=\!\mathbb{E}\!&\left\{\!\log\!\left(\!\! 1\!+\!\!\dfrac{|h_{22}h_{\tilde{1}1}|^2\!+\!|h_{11}h_{22}|^2}{|h_{11}h_{\tilde{1}2}|^2\!+\!|h_{12}h_{\tilde{1}1}|^2\!-\!2Re(h_{11}h_{\tilde{1}1}^*h_{\tilde{1}2}h_{12}^*)}\!\right)\!\!\right\}\!\!, \nonumber \\
F \!\!= \!\mathbb{E}\!&\left\{\log\left(1+\dfrac{|h_{12}|^2}{|h_{22}|^2}\right)\right\}, \nonumber
\end{align}
$(h_{\tilde{1}1},h_{\tilde{1}2})$ is identically distributed and independent of $(h_{11},h_{12})$.

\subsubsection{Gaps Between $R_\Sigma$ and Upper Bounds}
Next, we compare $U_1$, $U_2$ and $U_3$, with $R_{\Sigma}$ under the assumption that the topology variation is symmetric ($\lambda_{z_1}+\lambda_{z_4} = \lambda_{z_2}+\lambda_{z_3}$). Particularly, we consider two cases 1) $\lambda_{z_1} =\lambda_{z_2}$, \mbox{2) $\lambda_f \leq |\lambda_{z_1}-\lambda_{z_2}|$}.

\noindent Case 1): When $\lambda_{z_1}=\lambda_{z_2}$, and thus $\lambda_{z_3}=\lambda_{z_4}$, 
\begin{align}
&C_{\Sigma} \!-\! R_{\Sigma} \leq  U_1 \!-\! R_{\Sigma} \nonumber \\
= &(A\!+\!2B\!-\!C\!+E)(\lambda_{z_1}\!+\!\lambda_{z_4}) \!+\! \log(\pi e) \Bigg[\sum\limits_{k=1}^4{\lambda_{s_k}}\nonumber \\ 
&+\lambda_{m_1}+\lambda_{m_2}+2(\lambda_{i_1}+\lambda_{i_2})+2(\lambda_{z_1}+\lambda_{z_4})\Bigg].\label{eq:gap1}
\end{align}

\noindent Case 2): $U_2$ and $U_3$ are compared with $R_{\Sigma}$ respectively in each of following sub-cases:

i. When $\lambda_{z_1}<\lambda_{z_2}$, $\lambda_{z_3}<\lambda_{z_4}$, and $\lambda_f \leq \lambda_{z_2}-\lambda_{z_1}$,
\begin{align}
&C_{\Sigma} \!-\!\! R_{\Sigma}\leq  U_2 -R_{\Sigma} \nonumber \\
= & (A+2B-C)\left(\lambda_{z_1}+\lambda_{z_3}\right)+(4B-D)\lambda_f+\!F(\lambda_{z_2}\!+\!\lambda_{z_4}) \nonumber\\
&\!+\!\log(\pi e)\left[1\!+\!\lambda_{i_1}\!+\!\lambda_{i_2}\!+\!2(\lambda_{z_1}\!+\!\lambda_{z_4})\!+\!\lambda_f\right].\label{eq:gap2_1}
\end{align}

ii. When $\lambda_{z_1}> \lambda_{z_2}$, $\lambda_{z_3} > \lambda_{z_4}$, and $\lambda_f \leq \lambda_{z_1}-\lambda_{z_2}$,
\begin{align}
&C_{\Sigma} \!-\! R_{\Sigma}\leq  U_3 -R_{\Sigma} \nonumber \\
= & (A+2B-C)\left(\lambda_{z_2}+\lambda_{z_4}\right)+(4B-D)\lambda_f+ \!F(\lambda_{z_1}\!+\!\lambda_{z_3}) \nonumber \\
&+\log(\pi e)\left[1\!+\!\lambda_{i_1}\!+\!\lambda_{i_2}\!+\!2(\lambda_{z_1}\!+\!\lambda_{z_4})\!+\!\lambda_f\right].\label{eq:gap2_2}
\end{align}

\subsubsection{Gap Evaluation}
In the last step, we evaluate the obtained gaps of $C_{\Sigma}-R_{\Sigma}$ with the uniform-phase fading, Rayleigh fading and Rice fading respectively. 

\noindent\emph{Uniform-Phase Fading}

Because $|h_{ji}|^2=1$ for all $j,i$, $F=1$. We also have  
\begin{align}
A+2B-C =& \log(1\!+\!P)\!-\!\log\left(1\!+\!\dfrac{P}{2}\right)\!+\!2\log(1+2P) \nonumber \\
&-\! \log(P^2\!+\!3P\!+\!1) \nonumber\\
=& \log\dfrac{2P+2}{P+2} + \log\dfrac{4P^2+4P+1}{P^2+3P+1} \leq 3, \label{eq:A2BC}
\end{align}
and
\begin{align}
4B-D =& 2\left[2\log(2P+1)-\log\left(\dfrac{P^2}{2}+\dfrac{5P}{2} +1\right)\right] \nonumber \\
=& 2 \log\dfrac{8P^2+8P+2}{P^2+5P+2} \leq 6, \label{eq:4BD}
\end{align}

As we notice that $(h_{11}h_{\tilde{1}1}^*h_{\tilde{1}2}h_{12}^*)$ has magnitude 1, and its phase, mod $2\pi$, also admits a uniform distribution over $[0,2\pi)$, therefore we have $E = \dfrac{1}{2\pi}\int_0^{2\pi} \log\left( 1+\dfrac{1}{1-\cos \theta} \right)\,d\theta=1.9$.

When $\lambda_{z_1} = \lambda_{z_2}$, $\lambda_{z_3} = \lambda_{z_4}$, by (\ref{eq:gap1}) and (\ref{eq:A2BC}),
\begin{align}
C_{\Sigma} \!-\! R_{\Sigma} \leq&(3+1.9)(\lambda_{z_1}\!+\!\lambda_{z_4}) \!+\! \log(\pi e)\Bigg[\sum\limits_{k=1}^4{\lambda_{s_k}} \nonumber \\
&+\lambda_{m_1}+\lambda_{m_2}+2(\lambda_{i_1}+\lambda_{i_2})+2(\lambda_{z_1}+\lambda_{z_4})\Bigg] \nonumber\\
\leq& 4.9(\lambda_{z_1}+\lambda_{z_4})+ \log(\pi e)+\log(\pi e)(\lambda_{i_1}+\lambda_{i_2}) \nonumber\\
=&\log(\pi e) + \frac{4.9}{2}\left[\lambda_{i_1}+\lambda_{i_2}+2(\lambda_{z_1}+\lambda_{z_4})\right]\nonumber \\
&+\left(\log(\pi e)-\frac{4.9}{2}\right)(\lambda_{i_1}+\lambda_{i_2})\nonumber \\
& \leq 2\log(\pi e)\leq 6.2.
\end{align}

When $\lambda_f \leq |\lambda_{z_1}-\lambda_{z_2}|$, we can assume WLOG, that $\lambda_{z_1} > \lambda_{z_2}$, then by (\ref{eq:gap2_2}), (\ref{eq:A2BC}) and (\ref{eq:4BD}),
\begin{align}
C_{\Sigma} -R_{\Sigma} \leq & 3(\lambda_{z_2}\!+\!\lambda_{z_4})+6\lambda_f + \lambda_{z_1}+\lambda_{z_3}+2\log(\pi e) \nonumber \\
\leq &4(\lambda_{z_1}+\lambda_{z_3})+6.2 \leq 10.2.
\end{align}

\noindent \emph{Rayleigh Fading}

For i.i.d. Rayleigh fading, i.e., $h_{ik}$ are identically and independently Rayleigh distributed with $\mathbb{E}\{h_{ik}\}=0$, $\mathbb{E}\{|h_{ik}|^2\}=1$, for all $i \in \{1,\tilde{1},2\} $ and $k \in \{1,2\}$, we numerically evaluate the expectation of the expression $\log\!\left(\!\! 1\!+\!\dfrac{|h_{22}h_{\tilde{1}1}|^2\!+\!|h_{11}h_{22}|^2}{|h_{11}h_{\tilde{1}2}|^2\!+\!|h_{12}h_{\tilde{1}1}|^2\!-\!2Re(h_{11}h_{\tilde{1}1}^*h_{\tilde{1}2}h_{12}^*)}\!\right)$ to obtain an upper bound on $E$: $E \leq 1.457$. Similarly, we numerically evaluate the expectation of $\log\left(1+\dfrac{|h_{12}|^2}{|h_{22}|^2}\right)$ to obtain an upper bound on $F$:
\begin{equation*}
F = \int_{0}^{\infty} \!\!\! \int_{0}^{\infty}4xy \: e^{-(x^2+y^2)} \log\left(1+\dfrac{x^2}{y^2}\right) dx dy \leq 1.443.
\end{equation*}
We also numerically find that for arbitrary value of $P$, $A+2B-C \leq 4.35$, and $4B-D\leq 8.71$. From these, we calculate bounds on the gap in the same way as for the uniform-phase fading, and we have
\begin{equation}
C_{\Sigma}-R_{\Sigma} \leq \begin{cases}  6.2 \text{ bits/sec/Hz}, & \lambda_{z_1} = \lambda_{z_2}\\
12 \text{ bits/sec/Hz}, & \lambda_f \leq |\lambda_{z_1}-\lambda_{z_2}| \end{cases}.
\end{equation}

\noindent \emph{Rice Fading}

For i.i.d. Rice fading with Rice factor $K_r$, the channel coefficient $h_{ik}$, for all $i \in \{1,\tilde{1},2\} $ and $k \in \{1,2\}$, can be written as $h_{ik} = h_{Re} + jh_{Im}$ where $j^2=-1$, $h_{Re}\sim\mathcal{N}(\sqrt{K_r},\frac{1}{2})$ and $h_{Im}\sim\mathcal{N}(0,\frac{1}{2})$. For $K_r=1$, we numerically evaluate the terms $A+2B-C$, $4B-D$, $E$ and $F$, and find that for arbitrary value of $P$, $A+2B-C$ and $4B-D$ are upper bounded by 4.18 and 8.4 respectively. Also we have that $E \!=\! \mathbb{E} \! \left\{\!\log\!\left(\!\! 1\!\!+\!\!\dfrac{|h_{22}h_{\tilde{1}1}|^2\!+\!|h_{11}h_{22}|^2}{|h_{11}h_{\tilde{1}2}|^2\!+\!|h_{12}h_{\tilde{1}1}|^2\!-\!2Re(h_{11}h_{\tilde{1}1}^*h_{\tilde{1}2}h_{12}^*)}\!\!\right)\!\!\right\} \!\!\leq\! 1.67$, and $F = \mathbb{E}\left\{\log\left(1+\dfrac{|h_{12}|^2}{|h_{22}|^2}\right)\right\} \leq 1.39$. Applying these upper bounds to the derived gaps in (\ref{eq:gap1}), (\ref{eq:gap2_1}) and (\ref{eq:gap2_2}), we have for i.i.d. Rice fading with Rice factor $K_r=1$,
\begin{equation}
C_{\Sigma}-R_{\Sigma} \leq \begin{cases} 6.2\textup{ bits/sec/Hz}, &\lambda_{z_1}=\lambda_{z_2} \\ 11.77\textup{ bits/sec/Hz}, & \lambda_f \leq |\lambda_{z_1}-\lambda_{z_2}|\end{cases}.
\end{equation}

We emphasize here that our scheme is particularly well-suited to the high SNR regime since the gap becomes negligible with respect to the rate (i.e., ergodic capacity $\gg$ 10 bits/sec/Hz). 

\section{Numerical Results}
Having identified two coding opportunities where coding across topologies provide DoF gains, i.e., 50\% for coding opportunity 1 and 33.3\% for coding opportunity 2, compared with the TDMA scheme, we ask the question that how much rate gains can we actually get from each of the coding opportunities? In the first part of this section, we numerically evaluate the ergodic rate gains for Rayleigh and Rice fading channels in each of the coding opportunities. In the second part, we return to the rover-to-orbiter communication problem, and evaluate the DoF/throughput gains by coding across topologies for a 2-rover, 4-orbiter communication network.

\subsection{Ergodic Rate Gains of Coding Opportunities}
We numerically evaluate the ergodic rate gains over TDMA (one transmitter communicates with one receiver in each time slot) in both coding opportunities for the following i.i.d. fading channels: i. Rayleigh fading: $\mathbb{E}\{h_{ij}\}=0$, $\mathbb{E}\{|h_{ij}|^2\}=1$. ii. Rice fading: $\mathbb{E}\{h_{ij}\}=\sqrt{K_r}$, $\mathbb{E}\{|h_{ij}|^2\}=K_r+1$, where $K_r$ denotes the Rice factor.

\begin{figure}[htbp]
   \centering
   \subfigure[Coding opportunity 1.]{\includegraphics[width=0.48\textwidth]{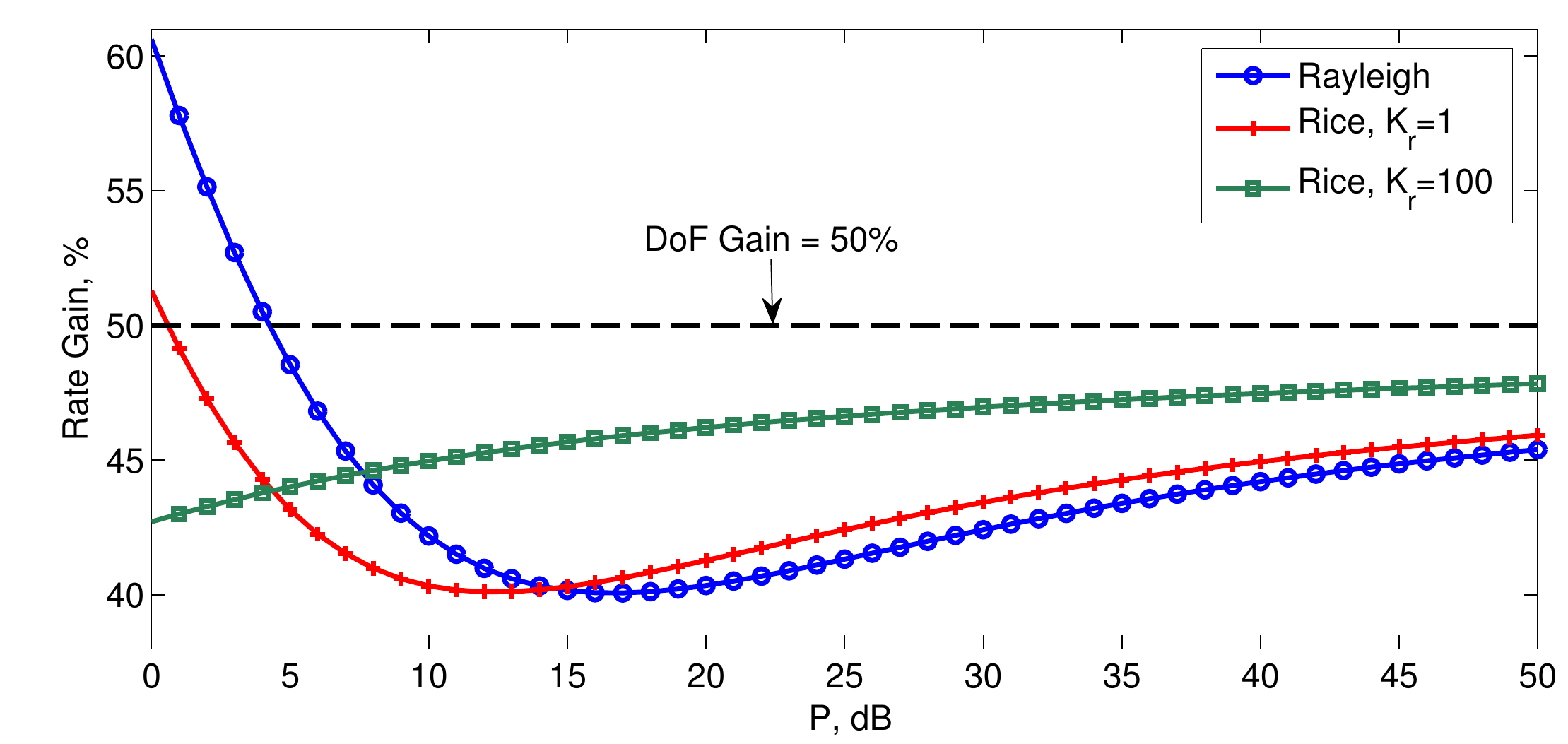}
       \label{fig:co1_rate}}
     \hfill
   \subfigure[Coding opportunity 2.]{\includegraphics[width=0.48\textwidth]{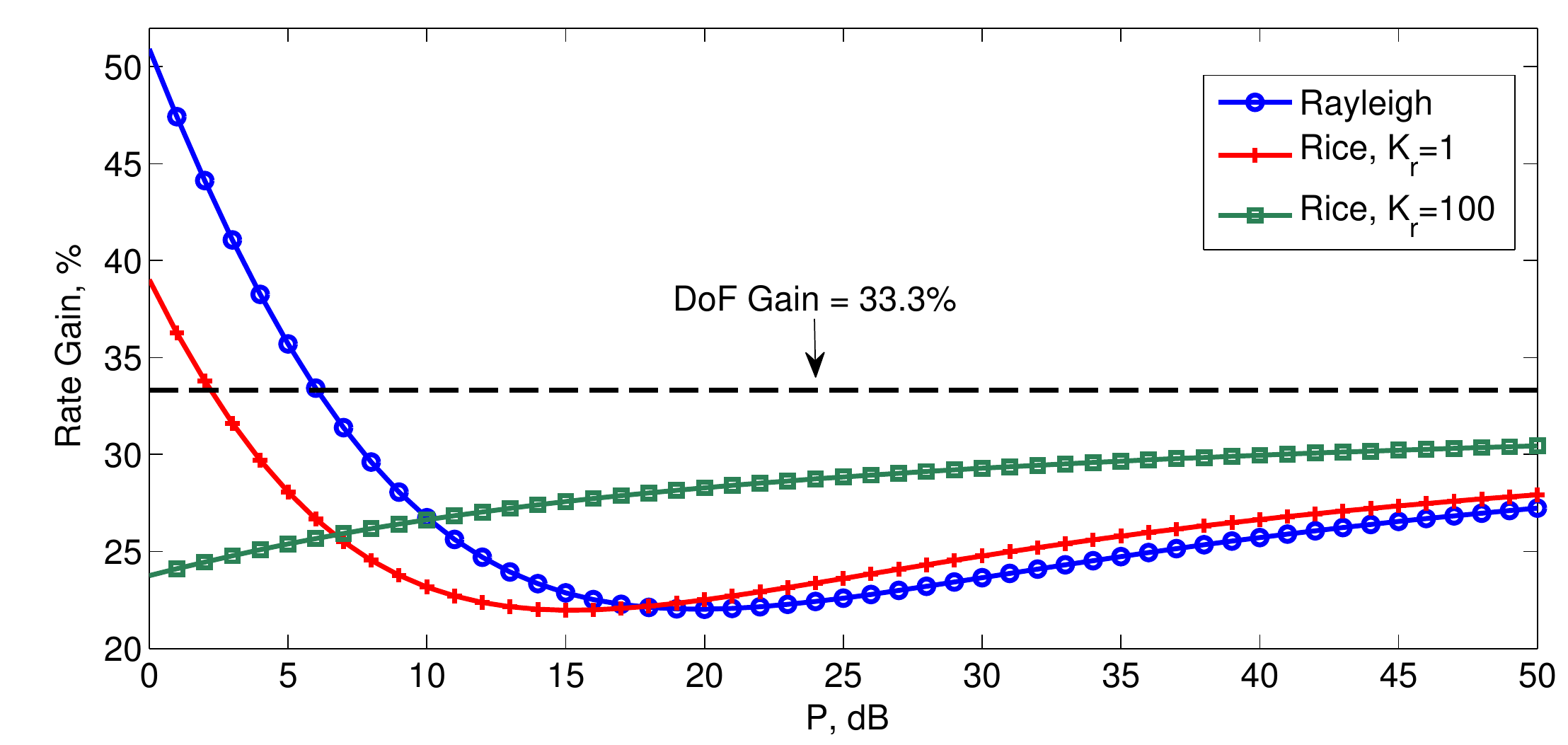}
       \label{fig:co2_rate}}
        \vspace{-2 mm}
   \caption{Ergodic rate gain by coding across topologies over TDMA.}
   \label{fig:co_rate}
\end{figure}

\vspace{-1mm}
\subsubsection{Coding Opportunity 1}
We assume that half of the topology instances are $z_1$, and the other half are $z_2$. We apply the proposed CAT scheme in Section~\ref{sec:coding1} to all $\{z_1,z_2\}$ combinations to numerically evaluate the ergodic sum-rate in (\ref{eq:ergodic_rate}), for the case $\lambda_{z_1}=\lambda_{z_2}=\frac{1}{2}$.

We see from Fig.~\ref{fig:co1_rate} that a substantial rate gain of at least 40\% can be achieved by the proposed CAT scheme over TDMA, for all channel distributions and transmit power constraints. As $P$ increases, the ergodic rate gain converges to the DoF gain of 50\%. 

\subsubsection{Coding Opportunity 2}
We assume that one third of the topology instances are $z_2$, another one third are $z_4$ and the remaining one third are $f$. We apply the proposed CAT scheme in Section~\ref{sec:coding2} to all $\{z_2,z_4,f\}$ combinations to numerically evaluate the ergodic sum-rate in (\ref{eq:ergodic_rate}), for the case $\lambda_{z_2}=\lambda_{z_4}= \lambda_{f}=\frac{1}{3}$. 

We see from Fig.~\ref{fig:co2_rate} that a rate gain of at least 22\% can be achieved by the proposed CAT scheme over TDMA, for all channel distributions and transmit power constraints. As $P$ increases, the ergodic rate gain converges to the DoF gain of 33.3\%.

Furthermore, we find that for both coding opportunities, when the transmit power constraint $P$ is small (e.g., less than 8 dB for coding opportunity 1 and less than 10 dB for coding opportunity 2), coding across topologies is more beneficial for the channels with many reflection/scattering multipath signal components. However, as $P$ increases, the rate gain by coding across topologies is more significant for the channels with a dominant signal component.

\subsection{Performance Evaluation of a 2-rover, 4-orbiter System}
For a 2-rover, 4-orbiter communication network, we numerically evaluate the DoFs and throughputs achieved by TDMA and our proposed scheme, to see how much performance gain our scheme can provide in practical settings. 

\subsubsection*{System Settings}
We use the satellite orbital parameters drawn from MRO and Odyssey. We assume that the four ascending nodes, which specify the longitudes of the orbits of the four orbiters, are uniformly spaced, and the two rovers are placed at the same latitude.  

We assume 500 K system temperature and 800 kHz system bandwidth. We employ the free space propagation model for the rover-to-orbiter communication link with path loss exponent 2. Each rover has transmit power of 10 W.

\begin{figure}[htbp]
   \centering
   \subfigure[DoF gain.]{\includegraphics[width=0.48\textwidth]{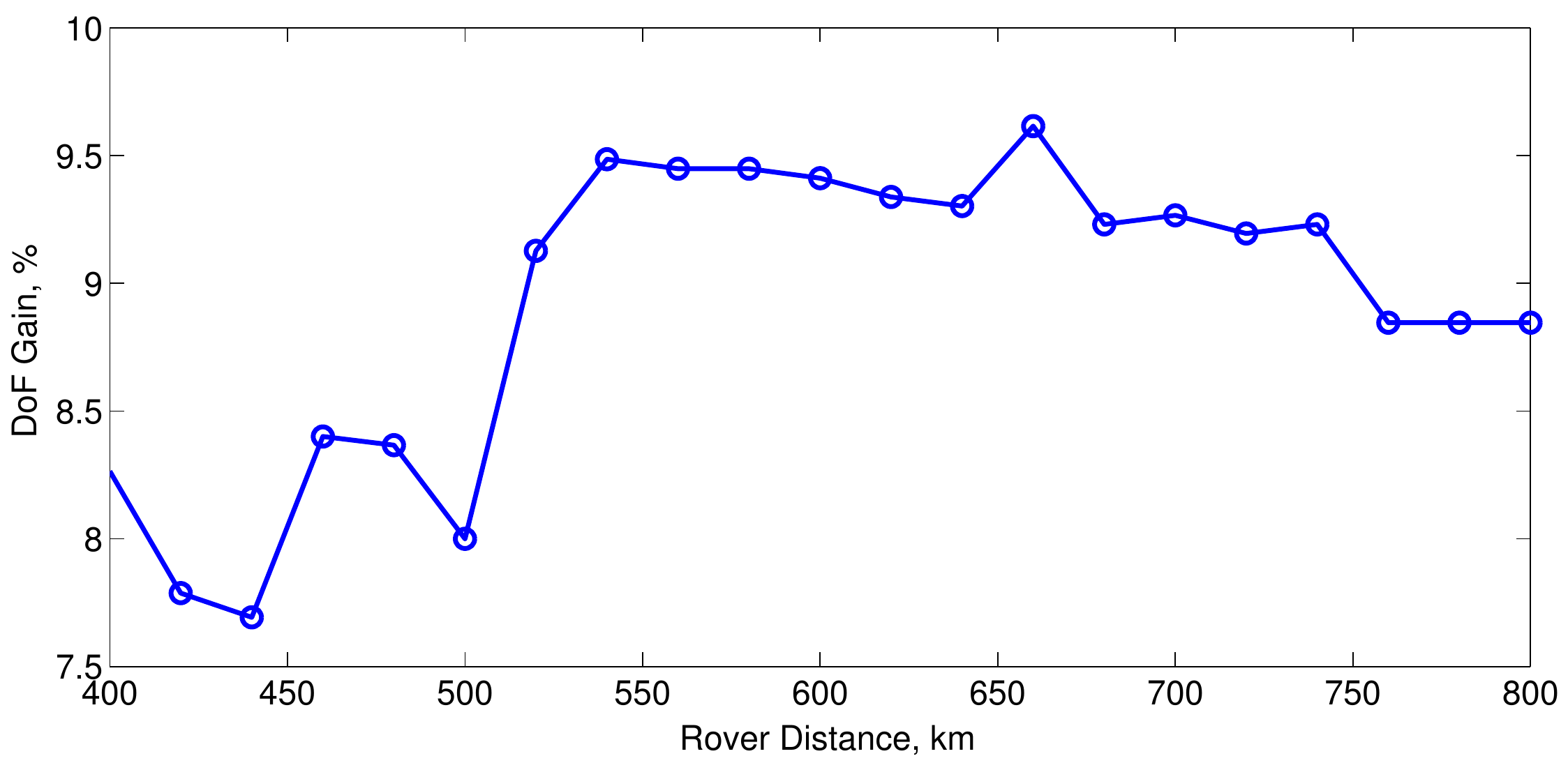}
       \label{fig:DoF_gain}}
     \hfill
   \subfigure[Throughput gain.]{\includegraphics[width=0.48\textwidth]{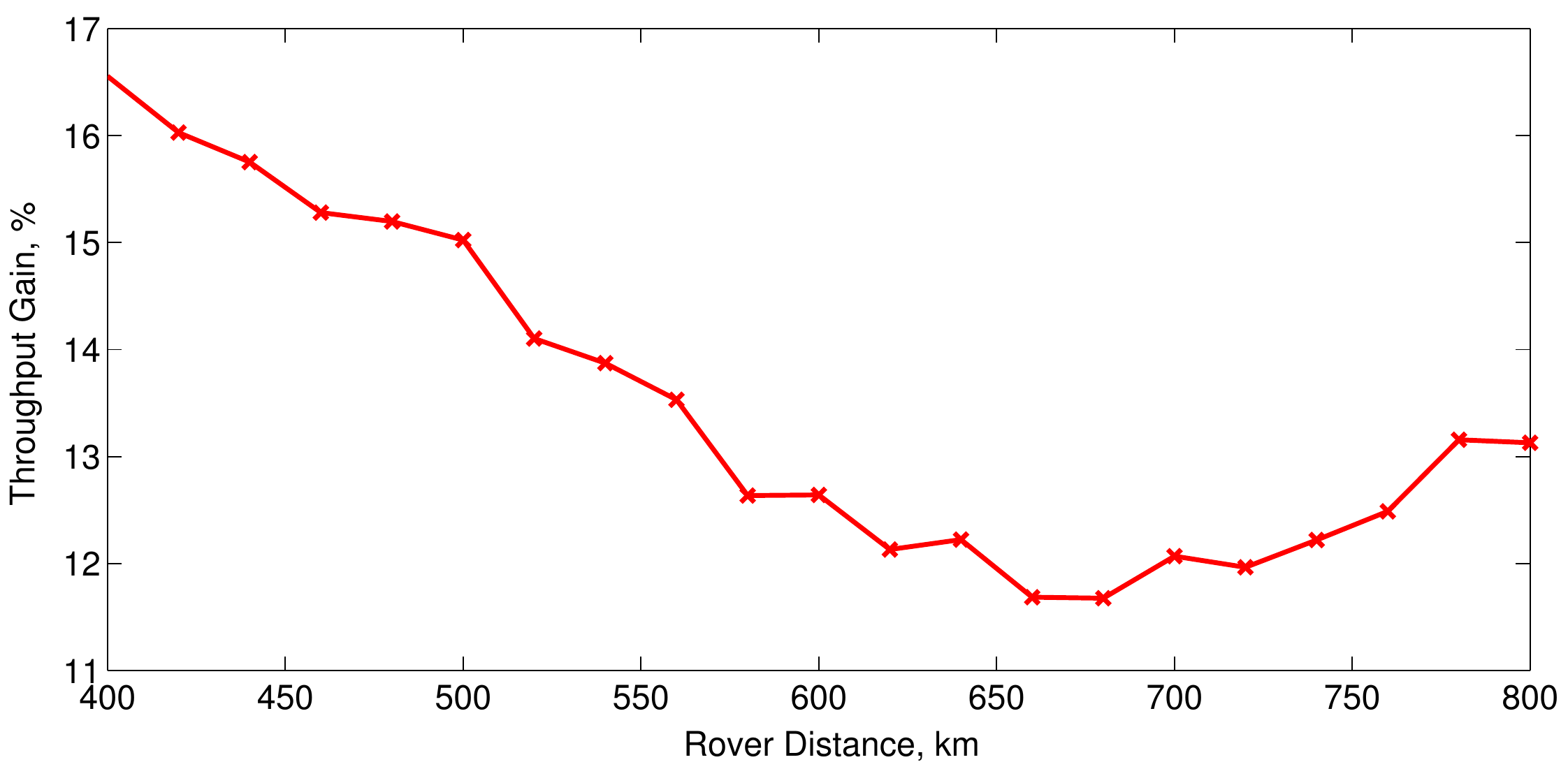}
       \label{fig:rate_gain}}
   \caption{Performance gains by applying coding across topologies to a 2-rover, 4-orbiter communication network.}
   \label{fig:performance}
\end{figure}

We find that, as shown in Fig.~\ref{fig:DoF_gain}, the DoF gain achieved by coding across topologies varies around 9\% across rover distances ranging from 400 km to 800 km, indicating that approximately 20\% of the topology instances participate in one of the coding opportunities. 

Fig.~\ref{fig:rate_gain} indicates that the proposed CAT scheme achieves at lease 11.6\% throughput gain over TDMA for all rover distances. We notice that the throughput gain consists of two parts: 1) the multiple-access rate gain by allowing two rovers to communicate to one orbiter simultaneously, 2) the rate gain by coding across topologies. When the rover distance increases from 400 km to around 650 km, while the DoF gain increases by 1\%, the throughput gain monotonically decreases. This is because the multiple-access rate gain drops more quickly within this distance range. After 650 km, the rate gain by coding across topologies starts to dominate, boosting up the overall throughput gain while the multiple-access rate gain keeps decreasing.   

\section{Concluding Remarks and Future Work}
In this paper, we modeled the communication problem from Mars rovers to Mars satellites as an X-channel with varying topologies, and conducted a theoretic analysis of the impact of coding across topologies on the DoF and the communication rate. For a $2 \times 2$ X-channel with varying topologies, we characterize the exact sum-DoF when the topology varies symmetrically and derive lower and upper bounds on the sum-DoF for the general asymmetric setting. Further, we numerically demonstrated the DoF/throughput gain in data rate provided by our proposed coding scheme over the baseline TDMA scheme that is currently in use. A natural extension of this work is to investigate methods/techniques (by either designing better coding schemes or proving a tighter upper bound) to characterize the exact sum-DoF without the symmetric constraint on the topology variation. Another interesting future direction is to explore the scenario with more than 4 orbiters (which is likely to happen in the future) to identify the potential new coding opportunities to further increase the DoF/throughput.

\section*{Appendix I\\ Proof of upper bounds in (\ref{eq:bond2}) and (\ref{eq:bond3})}
Recall that the output signals at Receiver~2 in Topology $a \in \mathcal{E}=\{b_1,b_2,f\}$, $Y_{2,a}^N$, is statistically equivalent to those received at Receiver~1, i.e., $Y_{1,a}^N$, and Receiver 2 decodes $(W_{21},W_{22})$ using
${Y_{2}^N \!\!=\!\! \left(\{Y_{2,s_k}^N\}_{k=1}^4, Y_{2,m_1}^N, Y_{2,m_2}^N,\{Y_{2,z_k}^N\}_{k=1}^4, Y_{2,i_1}^N, Y_{2,i_2}^N, Y_{1,\mathcal{E}}^N\right)}$.

To prove that (\ref{eq:bond2}) is an upper bound for the sum-DoF, we first use Fano's inequality at Receiver 2 and Receiver 1 respectively to establish upper bounds on $(R_{21}+R_{22})$ and $(R_{11}+R_{12})$. Then we simply add up these two bounds to obtain (\ref{eq:bond2}) as an upper bound on the sum-DoF.

At Receiver 2, by Fano's inequality and that $W_{11}$ is independent of $(W_{21},W_{22})$,
\begin{align}
&N (R_{21}+R_{22})\leq  I(W_{21},W_{22};Y_2^N|W_{11},{\boldsymbol h}^N)+N\epsilon_N \nonumber\\
= & h(Y_2^N|W_{11},{\boldsymbol h}^N)\!\!-\!\! h(Y_2^N|W_{21},W_{22},W_{11},{\boldsymbol h}^N) \!\!+\!\! N\! \epsilon_N\\
\leq & h(Y_2^N|W_{11},{\boldsymbol h}^N) \!-\! h(Y_{1,b_2}^N|W_{21},W_{22},W_{11},{\boldsymbol h}^N)\nonumber \\
&-\! h(Y_{2,z_2}^N|W_{21},W_{22},W_{11},Y_{1,b_2}^N,{\boldsymbol h}^N) \!+\! N \epsilon_N \label{eq:y2b2z2} \\
\overset{(a)}{=} & h(Y_2^N|W_{11},{\boldsymbol h}^N) - h(Y_{1,b_2}^N|W_{22},{\boldsymbol h}^N)\nonumber \\
&-h(Y_{2,z_2}^N|W_{21},W_{22},W_{11},X_{1,z_2}^N,Y_{1,b_2}^N,{\boldsymbol h}^N)\!+\! N\epsilon_N \\
\overset{(b)}{=} & h(Y_2^N|W_{11},{\boldsymbol h}^N) - h(Y_{1,b_2}^N|W_{22},{\boldsymbol h}^N) \nonumber\\
&-h\left((h_{22}X_2\!+\!Z_2)_{z_2}^N|W_{22},Y_{1,b_2}^N,{\boldsymbol h}^N\right) + N \epsilon_N \label{eq:EU2int1}\\
\overset{(c)}{\leq} & h(Y_2^N|W_{11},{\boldsymbol h}^N) - h(Y_{1,b_2}^N|W_{22},{\boldsymbol h}^N)+ No(\log P) \nonumber \\
&- h(Y_{1,z_2}^N|W_{22},Y_{1,b_2}^N,{\boldsymbol h}^N) + N\epsilon_N \\
\leq & h(Y_{2,s_2}^N|{\boldsymbol h}^N) \!+\! h(Y_{2,s_3}^N|{\boldsymbol h}^N)\!+\! h(Y_{2,m_2}^N|{\boldsymbol h}^N)\!+\! h(Y_{1,f}^N|{\boldsymbol h}^N) \nonumber \\
&+ \sum\limits_{k=1}^{2} h(Y_{2,i_k}^N|{\boldsymbol h}^N)\!\!+\! h(Y_{1,b_2}^N|{\boldsymbol h}^N)\!+\! \sum\limits_{k \neq 4}h(Y_{2,z_k}^N|{\boldsymbol h}^N) \nonumber \\
&\!+\!\! h(Y_{1,b_1}^N|W_{11},{\boldsymbol h}^N\!)\!\!+\! h(Y_{2,z_4}^N|W_{11},Y_{1,b_1}^N,{\boldsymbol h}^N)\!+\! N\! o(\log P)\nonumber\\
&- h(Y_{1,b_2}^N|W_{22},{\boldsymbol h}^N)\!-\! h(Y_{1,z_2}^N|W_{22},Y_{1,b_2}^N,{\boldsymbol h}^N)\!+\! N\epsilon_N, \label{eq:receiver2} 
\end{align}

Recall that $Y_{2,z_2}^N$ is the sub-vector of $Y_2^N$ in Topology $z_2$, and here $(h_{22}X_2\!+\!Z_2)_{z_2}^N$ denotes the sub-vector of the vector $\big\{h_{22}(n)X_2(n)+Z_2(n)\big\}_{n=1}^N$ in Topology $z_2$. Steps (a) and (b) result because $X_2^N$ is independent of $(W_{21},W_{11})$, and $X_1^N$ is a function of $(W_{21},W_{11})$. 

Step (c) holds because given $(h_{22}X_2\!+\!Z_2)_{z_2}^N$,
\begin{align}
h&\left(Y_{1,z_2}^N|(h_{22}X_2\!+\!Z_2)_{z_2}^N,{\boldsymbol h}^N\right) \nonumber\\
\leq & \sum_{\substack{n \in \text{time slots}\\
\text{in Topology } z_2}} \!\!\!\!h \bigg(Y_1(n)\!-\! \dfrac{h_{12}(n)\left(h_{22}(n)X_2(n)\!+\! Z_2(n)\right)}{h_{22}(n)} \bigg| \nonumber \\
& h_{22}(n)X_2(n) \!+\! Z_2(n),{\boldsymbol h}^N\bigg)\nonumber\\
 \leq &  \sum_{\substack{n \in \text{time slots}\\
\text{in Topology } z_2}} \!\!\! \mathbb{E}\left\{\log \left[\pi e\left(1+\left|\frac{h_{12}(n)}{h_{22}(n)}\right|^2\right)\right]\right\}\!=\! No(\log P),\label{eq:EU2int3}
\end{align}
and hence,
\begin{align}
-&h\left((h_{22}X_2\!+\!Z_2)_{z_2}^N|W_{22},Y_{1,b_2}^N,{\boldsymbol h}^N\right)\nonumber\\
=&  h\left(Y_{1,z_2}^N|(h_{22}X_2\!+\!Z_2)_{z_2}^N, W_{22},Y_{1,b_2}^N,{\boldsymbol h}^N\right)\nonumber \\
&-h\left((h_{22}X_2\!+\!Z_2)_{z_2}^N, Y_{1,z_2}^N|W_{22},Y_{1,b_2}^N,{\boldsymbol h}^N\right)\nonumber\\
\leq & No(\log P)\!-\! h\left(Y_{1,z_2}|W_{22},Y_{1,b_2}^N,{\boldsymbol h}^N\right). \label{eq:EU2int2}
\end{align}

Similarly at Receiver 1,
\begin{align}
N &(R_{11}+R_{12}) \nonumber\\
\leq & h(Y_1^N|W_{22},{\boldsymbol h}^N) \!-\! h(Y_{1,b_1}^N|W_{11},W_{12},W_{22},{\boldsymbol h}^N)\nonumber \\
&-\!h(Y_{1,z_4}^N|W_{11},W_{12},W_{22},Y_{1,b_1}^N,{\boldsymbol h}^N)\!+\!\! N\!\epsilon_N \label{eq:y1b1z4}\\
\leq & h(Y_{1,s_1}^N|{\boldsymbol h}^N) \!+\! h(Y_{1,s_4}^N|{\boldsymbol h}^N) \!+\! h(Y_{1,m_1}^N|{\boldsymbol h}^N)\!+\! h(Y_{1,f}^N|{\boldsymbol h}^N) \nonumber \\
&+ \sum\limits_{k=1}^{2} h(Y_{1,i_k}^N|{\boldsymbol h}^N)\!+\! h(Y_{1,b_1}^N|{\boldsymbol h}^N) + \sum\limits_{k \neq 2}h(Y_{1,z_k}^N|{\boldsymbol h}^N) \nonumber \\
&\!+\! h(Y_{1,b_2}^N|W_{22},{\boldsymbol h}^N)\!+\! h(Y_{1,z_2}^N|W_{22},Y_{1,b_2}^N,{\boldsymbol h}^N)\!+\! N \!o(\log P) \nonumber\\
&-h(Y_{1,b_1}^N|W_{11},{\boldsymbol h}^N\!)\!-\! \! h(Y_{2,z_4}^N|W_{11},Y_{1,b_1}^N,{\boldsymbol h}^N)\!+\! N \epsilon_N. \label{eq:receiver1}
\end{align}

Summing up (\ref{eq:receiver2}) and (\ref{eq:receiver1}), we have
\begin{align}
&N (R_{21}+R_{22}+R_{11}+R_{12})\nonumber\\
\leq& N \!\bigg[\!\sum\limits_{k=1}^{4}(\lambda_{s_k}\!\!+\!\lambda_{z_k}) \!+\! \sum\limits_{k=1}^{2}\left(\lambda_{b_k}\!+\!\lambda_{m_k}\!+\!2\lambda_{i_k}\right)\!+\! \lambda_{z_1} \!\!+\!\lambda_{z_3} \!\!+\!2\lambda_f \!\bigg]\nonumber \\
&(\log P \!+\! o(\log P))\!+\!\! No(\log P)\!+\!\! N \epsilon_N\nonumber\\
=&  N(1+\lambda_{i_1}+\lambda_{i_2}+\lambda_{z_1} +\lambda_{z_3}+\lambda_{f})(\log P + o(\log P))\nonumber \\
&+ No(\log P)\!+\! N\epsilon_N\label{eq:sumz13}.
\end{align}

Normalizing both sides of (\ref{eq:sumz13}) by $N \log(P)$, and letting both $P$ and $N$ go to infinity, we obtain (\ref{eq:bond2}).
 
Replacing $W_{11}$ with $W_{12}$, $Y_{1,b_2}^N$ with $Y_{1,b_1}^N$, and $Y_{2,z_2}^N$ with $Y_{2,z_3}^N$ in (\ref{eq:y2b2z2}), we obtain another bound for $N(R_{21}+R_{22})$. Replacing $W_{22}$ with $W_{21}$, $Y_{1,b_1}^N$ with $Y_{1,b_2}^N$, and $Y_{1,z_4}^N$ with $Y_{1,z_1}^N$ in (\ref{eq:y1b1z4}), we have another bound for $N(R_{11}+R_{12})$. Adding these two new upper bounds yields the sum-DoF upper bound in (\ref{eq:bond3}). 

\section*{Appendix II\\ Sketch of the Proof of Theorem~2}

\subsection*{Achievability}
The achievable scheme of Theorem~2 for the general $2 \times 2 $ X-channel with varying topologies (not necessarily symmetric) further extends the scheme of the symmetric setting presented in Section~III-C to include the following 2 additional cases:
\begin{itemize}
\item Case 3: $\lambda_{z_1} \leq \lambda_{z_2}$ and $\lambda_{z_3} > \lambda_{z_4}$,
\item Case 4: $\lambda_{z_1} > \lambda_{z_2}$ and $\lambda_{z_3} \leq \lambda_{z_4}$.
\end{itemize} 

Employing the same coding scheme as in Section III-C to exhaustively utilize the instances of coding opportunity~1 then the instances of coding opportunity~2 for Case 1 and 2, we can achieve the following sum-DoFs for the general asymmetric setting:
\begin{itemize}
\item Case 1: $\lambda_{z_1} \leq \lambda_{z_2}$ and $\lambda_{z_3} \leq \lambda_{z_4}$
\begin{align}
d_{\Sigma} \geq& 1+\lambda_{i_1}+\lambda_{i_2} \nonumber \\
&+\min\{\lambda_{z_1}+\lambda_{z_4},\lambda_{z_2}+\lambda_{z_3},\lambda_{z_1}+\lambda_{z_3}+\lambda_f\},
\end{align}
\item Case 2: $\lambda_{z_1} > \lambda_{z_2}$ and $\lambda_{z_3} > \lambda_{z_4}$
\begin{align}
d_{\Sigma} \geq& 1+\lambda_{i_1}+\lambda_{i_2} \nonumber \\
&+\min\{\lambda_{z_1}+\lambda_{z_4},\lambda_{z_2}+\lambda_{z_3},\lambda_{z_2}+\lambda_{z_4}+\lambda_f\},
\end{align}
\end{itemize} 

For Case 3 and 4, after all possible instances of coding opportunity 1 are utilized (as in Step i of the scheme in Section III-C), delivering $3\min\{\lambda_{z_1},\lambda_{z_2}\}N+3\min\{\lambda_{z_3},\lambda_{z_4}\}N$ symbols, there is no coding opportunity 2 to take advantage of and all the remaining topologies are coded separately to deliver one symbol at each channel use. Therefore, the achieved sum-DoF is $1+\lambda_{i_1}+\lambda_{i_2}+\lambda_{z_1}+\lambda_{z_4}$ for Case~3 and $1+\lambda_{i_1}+\lambda_{i_2}+\lambda_{z_2}+\lambda_{z_3}$ for Case~4.

Overall, the achieved sum-DoF without the symmetric constraint is:
\begin{align}
&1+\lambda_{i_1}+\lambda_{i_2} \nonumber \\
&+\min\{\lambda_{z_1}\!+\! \lambda_{z_4},\lambda_{z_2}\!+\!\lambda_{z_3},\lambda_{z_1}\!+\!\lambda_{z_3}\!+\!\lambda_f,\lambda_{z_2}\!+\!\lambda_{z_4}\!+\!\lambda_f\}.
\end{align}

\subsection*{Upper Bound}
The converse of Theorem~2 includes the following 3 upper bounds on the sum-DoF of the general asymmetric setting: 

\begin{align}
d_{\Sigma} &\leq 1+\lambda_{i_1}+\lambda_{i_2}+ \frac{\lambda_{z_1}+\lambda_{z_2}+\lambda_{z_3}+\lambda_{z_4}}{2}, \label{eq:Abond1}\\
d_{\Sigma} &\leq 1+\lambda_{i_1}+\lambda_{i_2}+\lambda_{z_1}+\lambda_{z_3}+\lambda_{f},\label{eq:Abond2}\\
d_{\Sigma} &\leq 1+\lambda_{i_1}+\lambda_{i_2}+\lambda_{z_1}+\lambda_{z_3}+\lambda_{f},\label{eq:Abond3}
\end{align} 
where (\ref{eq:Abond2}) and (\ref{eq:Abond3}) are the same as (\ref{eq:bond2}) and (\ref{eq:bond3}) respectively and hold no matter if the topology variation is symmetric or not (see proofs in Appendix~I). Here we demonstrate how the upper bound (\ref{eq:Abond1}) is derived.

Following the same steps as proving the upper bound (\ref{eq:bond1}) for the symmetric setting until (\ref{eq:w1ty1b}), we have for Receivers 1 and $\tilde{1}$:
\begin{align}
&I (W_1;Y_{1,\mathcal{B}}^N|\tilde{{\boldsymbol h}}^N)+ I(W_1;Y_{\tilde{1},\mathcal{B}}^N|\tilde{{\boldsymbol h}}^N)\nonumber\\
\leq & 2N(\lambda_{i_1}+\lambda_{i_2}+\lambda_{z_1}+\lambda_{z_4})(\log P + o(\log P))\nonumber \\
&+2I(W_1;Y_{1,z_2}^N,Y_{1,z_3}^N,Y_{1,\mathcal{E}}^N|\tilde{{\boldsymbol h}}^N)\nonumber \\
&-h(Y_{1,z_1}^N,Y_{1,z_4}^N,Y_{\tilde{1},z_1}^N,Y_{\tilde{1},z_4}^N|W_1,Y_{1,z_2}^N,Y_{1,z_3}^N,Y_{1,\mathcal{E}}^N,\tilde{{\boldsymbol h}}^N)\\
= & 2N(\lambda_{i_1}+\lambda_{i_2}+\lambda_{z_1}+\lambda_{z_4})(\log P + o(\log P)) \nonumber \\
&+\! 2I(W_1;Y_{1,z_2}^N,Y_{1,z_3}^N,Y_{1,\mathcal{E}}^N|\tilde{{\boldsymbol h}}^N\!)\nonumber \\
&+ h(Y_{1,z_2}^N,Y_{1,z_3}^N,Y_{1,\mathcal{E}}^N|W_1,\tilde{{\boldsymbol h}}^N\!)\nonumber \\
&-h(Y_{1,z_1}^N,Y_{1,z_4}^N,Y_{\tilde{1},z_1}^N,Y_{\tilde{1},z_4}^N,Y_{1,z_2}^N,Y_{1,z_3}^N,Y_{1,\mathcal{E}}^N|W_1,\tilde{{\boldsymbol h}}^N)\\
\leq & 2N(\lambda_{i_1}+\lambda_{i_2}+\lambda_{z_1}+\lambda_{z_4})(\log P + o(\log P)) \nonumber \\
&+ 2I(W_1;Y_{1,z_2}^N,Y_{1,z_3}^N,Y_{1,\mathcal{E}}^N|\tilde{{\boldsymbol h}}^N) \nonumber \\
&+ h(Y_{1,z_2}^N,Y_{1,z_3}^N,Y_{1,\mathcal{E}}^N|W_1,\tilde{{\boldsymbol h}}^N) \nonumber \\
&-h(Y_{1,z_1}^N,Y_{1,z_4}^N,Y_{\tilde{1},z_1}^N,Y_{\tilde{1},z_4}^N,Y_{1,z_2}^N,Y_{1,z_3}^N,Y_{1,\mathcal{E}}^N|W_1,\tilde{{\boldsymbol h}}^N)\nonumber\\
& + h(Y_{1,z_1}^N,Y_{1,z_4}^N,Y_{\tilde{1},z_1}^N,Y_{\tilde{1},z_4}^N,Y_{1,z_2}^N,Y_{1,z_3}^N,Y_{1,\mathcal{E}}^N|W_1,W_2,\tilde{{\boldsymbol h}}^N)\nonumber \\
&-h(Y_{1,z_2}^N,Y_{1,z_3}^N,Y_{1,\mathcal{E}}^N|W_1,W_2,\tilde{{\boldsymbol h}}^N)\nonumber\\
= & 2N(\lambda_{i_1}+\lambda_{i_2}+\lambda_{z_1}+\lambda_{z_4})(\log P + o(\log P)) \nonumber \\
&+ 2I (W_1;Y_{1,z_2}^N,Y_{1,z_3}^N,Y_{1,\mathcal{E}}^N|\tilde{{\boldsymbol h}}^N) \nonumber \\
&+ I(W_2;Y_{1,z_2}^N,Y_{1,z_3}^N,Y_{1,\mathcal{E}}^N|W_1,\tilde{{\boldsymbol h}}^N)\nonumber \\
&-\! I(W_2;Y_{1,z_1}^N,Y_{1,z_4}^N,Y_{\tilde{1},z_1}^N,Y_{\tilde{1},z_4}^N,Y_{1,z_2}^N,Y_{1,z_3}^N,Y_{1,\mathcal{E}}^N|W_1,\tilde{{\boldsymbol h}}^N\!)\nonumber\\
\leq & 2N(\lambda_{i_1}+\lambda_{i_2}+\lambda_{z_1}+\lambda_{z_4})(\log P + o(\log P)) \nonumber \\
&+ I(W_1,W_2;Y_{1,z_2}^N,Y_{1,z_3}^N,Y_{1,\mathcal{E}}^N|\tilde{{\boldsymbol h}}^N) \nonumber \\
&+I(W_1;Y_{1,z_2}^N,Y_{1,z_3}^N,Y_{1,\mathcal{E}}^N|\tilde{{\boldsymbol h}}^N) \nonumber \\
&- I(W_2;Y_{1,z_1}^N,Y_{1,z_4}^N,Y_{\tilde{1},z_1}^N,Y_{\tilde{1},z_4}^N,Y_{1,\mathcal{E}}^N|W_1,\tilde{{\boldsymbol h}}^N)\\
\overset{(a)}{\leq} & N\left[2(\lambda_{i_1}\!+\!\lambda_{i_2}\!+\!\lambda_{z_1}\!+\!\lambda_{z_4})\!+\!\lambda_{z_2}\!+\!\lambda_{z_3}\!+\!\lambda_{b_1}\!+\!\lambda_{b_2}\!+\!\lambda_f\right]\nonumber \\&\left(\log P+o(\log P)\right)+ I(W_1;Y_{1,z_2}^N,Y_{1,z_3}^N,Y_{1,\mathcal{E}}^N|W_2,\tilde{{\boldsymbol h}}^N )\nonumber \\
&- I(W_2;Y_{2,z_1}^N, Y_{2,z_4}^N,Y_{1,\mathcal{E}}^N|W_1,\tilde{{\boldsymbol h}}^N )+No(\log P),\label{eq:Aint1}
\end{align}
where (a) holds following the same derivation from (\ref{eq:EUint3}) to (\ref{eq:w1y1bw1y1t}) in Section~IV except that no symmetric constraint is imposed on the topology variation.

Similarly for Receivers 2 and $\tilde{2}$, we have  
\begin{align}
I&(W_2;Y^N_{1,\mathcal{E}},Y_{2,\mathcal{B}\backslash\mathcal{E}}^N|\tilde{{\boldsymbol h}}^N )+ I(W_2;Y^N_{1,\mathcal{E}},Y_{\tilde{2},\mathcal{B}\backslash\mathcal{E}}^N|\tilde{{\boldsymbol h}}^N )\nonumber\\
\leq & N\left[2(\lambda_{i_1}\!+\! \lambda_{i_2}\!+\! \lambda_{z_2}\!+\! \lambda_{z_3})\!+\! \lambda_{z_1}\!+\! \lambda_{z_4}\!+\! \lambda_{b_1}\!+\! \lambda_{b_2}\!+\! \lambda_{f}\right]\nonumber \\
&\left(\log P+o(\log P)\right) \nonumber + I(W_2;Y_{2,z_1}^N, Y_{2,z_4}^N,Y_{1,\mathcal{E}}^N|W_1,\tilde{{\boldsymbol h}}^N ) \nonumber \\
&- I(W_1;Y_{1,z_2}^N,Y_{1,z_3}^N,Y_{1,\mathcal{E}}^N|W_2,\tilde{{\boldsymbol h}}^N )+No(\log P).\label{eq:Aint2}
\end{align}

Adding both sides of (\ref{eq:R1}) and (\ref{eq:R2}) in Section~IV, and by (\ref{eq:Aint1}) and (\ref{eq:Aint2}), we arrive at
\begin{align}
2&N(R_1+R_2) \nonumber \\
\leq & N\!\left[2 \! \left( \! 2 \lambda_{i_1} + 2\lambda_{i_2} + \sum\limits_{k=1}^4 \lambda_{z_k}+\lambda_{b_1}+\lambda_{b_2}+\lambda_f \! \right) \!+\! \sum\limits_{k=1}^4 \lambda_{z_k}\right] \nonumber \\
&(\log P+o(\log P))\!+\!No(\log P)\!+\! N \epsilon_N\nonumber\\
& +2N\left(\sum\limits_{k=1}^4 \lambda_{s_k} + \lambda_{m_1} + \lambda_{m_2}\right)(\log P+o(\log P))  \nonumber\\
=& N \! \left[2\left(1+\lambda_{i_1}+\lambda_{i_2}\right)+ \sum\limits_{k=1}^4 \lambda_{z_k}\right] \! (\log P+o(\log P))\nonumber \\
& + \! No(\log P)\!+\! N \epsilon_N\label{eq:Aint3}.
\end{align}

Dividing both sides of (\ref{eq:Aint3}) by $2N \log P$ and let both $N$ and $P$ go to infinity, we obtain the sum-DoF upper bound in (\ref{eq:Abond1}).

\section*{Acknowledgement}
The authors would like to thank Dr. Dariush Divsalar at JPL for useful discussions. 

\bibliographystyle{IEEEtran}
\bibliography{ref}

\end{document}